\documentclass[lettersize,journal,10pt,twocolumn]{IEEEtran}

\usepackage{booktabs}
\usepackage{wrapfig}
\usepackage{lipsum}
\usepackage{verbatim}
\usepackage{algorithm}
\usepackage{algorithmicx}
\usepackage{algpseudocode}
\usepackage{graphicx}
\usepackage{cite}
\usepackage{amsmath}
\usepackage{amsthm}
\usepackage{epstopdf}
\usepackage[font=small,labelfont=bf]{caption}
\usepackage{makecell}
\usepackage{arydshln}
\usepackage{multirow}
\DeclareUnicodeCharacter{2113}{\ensuremath{\ell}}
\usepackage{multicol}

\usepackage{nicefrac}
\usepackage{bm}
\usepackage{amsfonts}
\usepackage{float}

\usepackage{color}

\usepackage{mdwmath}
\usepackage{url}
\usepackage{amssymb}

\usepackage{tikz}
\usetikzlibrary{
    positioning,
    arrows.meta,
    shapes.geometric,
    fit,
    calc
}
\usepackage{adjustbox} % For scaling the diagram to fit within the page

\newtheorem{lemma}{\textbf{Lemma}}
\newtheorem{theorem}{\textbf{Theorem}}

\newtheorem{corollary}{\textbf{Corollary}}

\usepackage{hyperref}

\def\0{{\mathbf 0}}
\def\1{{\mathbf 1}}

\def\a{{\mathbf a}}

\def\f{{\mathbf f}}
\def\g{{\mathbf g}}
\def\h{{\mathbf h}}

\def\m{{\mathbf m}}

\def\r{{\mathbf r}}
\def\s{{\mathbf s}}

\def\u{{\mathbf u}}
\def\v{{\mathbf v}}

\def\x{{\mathbf x}}
\def\y{{\mathbf y}}

\def\D{{\mathbf D}}

\def\H{{\mathbf H}}
\def\I{{\mathbf I}}

\def\L{{\mathbf L}}
\def\M{{\mathbf M}}

\def\P{{\mathbf P}}
\def\Q{{\mathbf Q}}
\def\R{{\mathbf R}}

\def\V{{\mathbf V}}
\def\W{{\mathbf W}}
\def\X{{\mathbf X}}

\def\Z{{\mathbf Z}}

\def\ie{{\textit{i.e.}}}
\def\eg{{\textit{e.g.}}}

\def\cC{{\mathcal C}}
\def\cD{{\mathcal D}}
\def\cE{{\mathcal E}}

\def\cG{{\mathcal G}}

\def\cI{{\mathcal I}}

\def\cN{{\mathcal N}}
\def\cO{{\mathcal O}}

\def\cR{{\mathcal R}}
\def\cS{{\mathcal S}}

\def\bLambda{{\boldsymbol \Lambda}}

\ifCLASSOPTIONcompsoc 
 \usepackage[caption=false,font=normalsize,labelfont=sf,textfont=sf]{subfig}
 \else
 \usepackage[caption=false,font=footnotesize]{subfig}
 \fi
\begin{document}

%
% The "title" command has an optional parameter, allowing the author to define a "short title" to be used in page headers.
%\title{\huge Joint Time-varying Graph Estimation / Signal Interpolation via Smoothness and Low-Rank Priors}
\title{Low-rank Updates in Slowly Time-varying Graphs for Spatial-Temporal
 Signal Interpolation}
\author{%\IEEEauthorblockN{abc}
\IEEEauthorblockN{Saghar Bagheri, Gene Cheung, \emph{Fellow, IEEE}, Tim Eadie, Antonio Ortega, \emph{Fellow, IEEE}}
\renewcommand{\baselinestretch}{1.0}
\thanks{The work of G. Cheung was supported in part by the Natural Sciences and Engineering Research Council of Canada (NSERC) RGPIN-2025-06252. \emph{(Corresponding author: Gene Cheung.)}}
\thanks{{S. Bagheri, G. Cheung, and T. Eadie are with the Department of EECS, York University, Toronto, Canada (e-mail: baqeri.saqar@gmail.com, \{genec, vorin\}@yorku.ca). A. Ortega is with the Dept of ECE, University of Southern California, Los Angeles, CA. (email: aortega@usc.edu).}}
}
\maketitle
%
% The abstract is a short summary of the work to be presented in the article.
\begin{abstract}
A crucial assumption in graph signal processing (GSP) is the existence of an underlying graph that captures the pairwise similarities between nodes, allowing filters to be designed based on this graph for tasks such as denoising.
For spatial-temporal data in which node-to-node similarities evolve over time, a static spatial graph is insufficient.
In this paper, to represent slowly time-varying pairwise relationships, we model the graph changes in two consecutive adjacency matrices $\P = \W^{(2)} - \W^{(1)}$ across time as a low-rank matrix.  
Specifically, given an initial adjacency matrix $\W^{(1)}$ at time $t=1$, we jointly interpolate a signal $\x_2$ and estimate $\W^{(2)}$ at $t=2$ using both a graph signal smoothness prior for $\x_2$ and a low-rank prior on $\P$.
We alternate optimization steps. 
With $\W^{(2)}$ fixed, $\x_2$ is interpolated by solving a linear system. 
Alternatively, holding $\x_2$ fixed, $\W^{(2)}$ is updated via proximal gradient descent (PGD). 
The proximal mapping of the rank term $\Gamma(\W^{(2)} - \W^{(1)})$ is approximated in linear time using a fast orthogonal matching pursuit (OMP) algorithm that selects a sparse combination of atoms from a dictionary $\cR$ formed by the outer products of $\W^{(1)}$'s eigenvectors.
We unroll iterations of our algorithm into layers to build a lightweight neural network for limited data-driven parameter tuning.
Experiments show that our joint optimization achieves better signal interpolation compared to existing time-varying graph models.

\end{abstract}
%
% Keywords. The author(s) should pick words that accurately describe the work being
% presented. Separate the keywords with commas.
\begin{IEEEkeywords}
Time-varying graphs, graph signal interpolation, low-rank matrices, perturbation theory, algorithm unrolling
\end{IEEEkeywords}

%\begin{teaserfigure}
%  \includegraphics[width=\textwidth]{sampleteaser}
%  \caption{Seattle Mariners at Spring Training, 2010.}
%  \Description{Enjoying the baseball game from the third-base seats. Ichiro Suzuki preparing to bat.}
%  \label{fig:teaser}
%\end{teaserfigure}

%
% This command processes the author and affiliation and title information and builds
% the first part of the formatted document.
\maketitle

\section{Introduction}
\label{sec:intro}
\textit{Graph Signal Processing} (GSP) \cite{ortega18ieee,cheung18} studies computational tools for discrete signals residing on finite graphs, which serve as structured data kernels that capture pairwise similarities or correlations between nodes \cite{giannakis18,dong19}.
A statically defined graph\footnote{Though spectral analysis of \textit{signed} graphs with both positive and negative edges has been studied recently \cite{yang22,Dinesh2025}, we focus on the more common graphs with \textit{positive} edge weights in this paper.} allows filters to leverage defined one-hop neighborhoods for local filtering tasks, such as denoising, dequantization, and interpolation \cite{pang17,liu17,chen24}. However, 
For \textit{spatial-temporal data}, such as traffic and weather, these node-to-node similarities often evolve gradually over time. 
For example, crop yields of neighboring counties in the corn belt may become increasingly similar as farmers adopt common cultivation practices \cite{prasad06,Cai17}. 
%and land prices in adjacent counties may also converge as local markets respond to mutual cues.
Thus, the graph model should also be updated periodically to reflect these evolving pairwise similarity relationships. 
In this paper, we focus on learning \textit{slowly time-varying graphs} from spatial-temporal data.

Existing methods for modeling temporal changes in graphs can be grouped into three categories (discussed in detail in \autoref{sec:related}).
Some approaches model ``small" temporal graph changes via structural assumptions in the \textit{nodal} domain; for example, \cite{yamada19} assumes sparsity in graph topology updates (\ie, changes in only a few edges), while  \cite{kalofolias17} assumes small topology changes (\ie, changes in the adjacency matrix) in Frobenius norm over time.
These assumptions help reduce the amount of data needed to update the model parameters.
They differ mainly in how to mathematically describe ``small'' changes in slowly time-varying graphs. 

% \begin{table}[t!]
%  \caption{When $\W^{(t+1)}$ differs from $\W^{(t)}$ by a random low-rank matrix, approximating the matrix using outer-products of $\W^{(t)}$'s eigenvectors is better than outer-products of random vectors, for various graph sizes $N$. (Errors are the difference between reconstructed and true $\W^{(t+1)}$ in Frobenius norm.) }
% \label{tab:EV_outer}
%     \centering
%     \begin{tabular}{c|ccccc}
%     approx. scheme & $N=50$ & $70$ & $100$ & $150$ & $200$ \\ \hline
%     random vectors & 0.988  & 0.992 & 0.994 & 0.996 & 0.997 \\
%     e-vectors & 0.355 & 0.309 & 0.253 & 0.216 & 0.191 
%     \end{tabular}
% \end{table}

In contrast to \cite{yamada19} and \cite{kalofolias17}, to capture slowly evolving time-varying dynamics in spatial-temporal data, we assume that the \textit{difference} between two consecutive estimated adjacency matrices, $\P = \W^{(t+1)} - \W^{(t)}$, is \textit{low-rank}. 
%$\W^{(t)}$ and $\W^{(t+1)}$ at two consecutive time instants $t$ and $t+1$ 
%is low-rank. 
This is a good assumption \textit{when only a few nodes change behavior}; for example, a county adopting new tax-relief regulations could lead to higher farmland valuations than its neighbors---a rank-2 change (\ie, one row/column update).
Note that this rank-2 update matrix $\P$ may not be sparse as assumed in \cite{yamada19} or small in Frobenius norm as assumed in \cite{kalofolias17}.

While the low-rank prior has been used in applications such as recommendation systems \cite{sarwar02} and community detection \cite{tong08}, it has not been used to model the changes between two consecutive similarity graphs\footnote{\cite{sarwar02,tong08,kanada18} assumed that the adjacency matrix itself is low-rank, which was empirically shown in \cite{bagheri24} to be a poorer statistical fit for typical connected weighted graphs in observed datasets without obvious $k$-partite structures.}. 
%which we verify empirically using real-world data. 
%It can be shown experimentally, 
Compared to works that restrict changes in nodal or spectral-domain structures \cite{chentong16,Li17,cao15,kalofolias17,yamada19}, it has been experimentally shown that the low-rank prior remains effective even with limited observed data \cite{candes11}.

Next, given that the graph update matrix $\P$ is low-rank\footnote{\cite{Fernandez2025} studies rank-1 updates of Laplacian matrices of path graphs, while we study low-rank updates of adjacency matrices of general positive graphs.}, we show that it can be well approximated by a sparse linear combination of outer-products of eigenvectors of $\W^{(t)}$:
\begin{align}
 \P \approx \sum_{i \in \cI} a_i \v_i \v_i^\top + b_i (\v_1 \v_i^\top + \v_i \v_1^\top)  
\end{align}
where $\v_i$ is the eigenvector corresponding to the $i$-th largest eigenvalue $\lambda_i$ of $\W^{(t)}$, and $\cI$ is an index set to denote selected rank-1 matrices from dictionary 
$\cD = \{\v_i \v_i^\top, \v_1 \v_i^\top, \v_i \v_1^\top \}_{i=1}^N$. 
%As illustration, assuming that $\W^{(t+1)}$ differs from $\W^{(t)}$ by a random low-rank matrix, Table\;\ref{tab:EV_outer} shows that approximating $\W^{(t+1)} - \W^{(t)}$ using outer-products of eigenvectors of $\W^{(t)}$ for low-rank updates results in much smaller errors compared to outer-products of random vectors, especially for large graph sizes. 
In so doing, \textit{we are replacing a low-rank matrix estimation problem with a less expensive sparse coding problem.}
One intuitive explanation is that the \textit{Perron vector} $\v_1$ corresponding to the largest eigenvalue $\lambda_1$ captures node importance (\textit{eigenvector centrality} \cite{bonacich87}), so the outer-product $\v_1 \v_1^\top$ captures interactions between important nodes, while $\v_1 \v_i^\top$ captures interactions between important nodes and other nodes.
In more mathematically rigorous terms, we derive a proof for this sparse coding approximation via the \textit{Concentration of Measure Phenomenon} (CoMP) \cite{Ledoux2001} in Section\;\ref{subsec:justify}.

To better leverage limited signal observations, previous GSP works \cite{liu17,bai19} have shown that \textit{jointly} estimating a graph and restoring its associated signal can improve overall recovery performance. 
%compared to treating those steps independently.
Inspired by these works, observing only partial signal $\y_2$ at time $t=2$, we jointly interpolate signal $\x_2$ from $\y_2$ \textit{and} estimate the updated adjacency matrix $\W^{(2)}$, given signal $\x_1$ and an adjacency matrix $\W^{(1)}$ at time $t=1$. 
Our model incorporates both a graph signal smoothness prior---the quadratic \textit{graph Laplacian regularizer} (GLR) \cite{pang17} for $\x_2$—to capture spatial similarities, and a low-rank prior for $\P = \W^{(2)}-\W^{(1)}$ to model temporal changes in spatial similarities.
We alternate between computing $\x_2$ by solving a linear system when $\W^{(2)}$ is fixed, and updating $\W^{(2)}$ via \textit{proximal gradient descent} (PGD) \cite{parikh13} when $\x_2$ is fixed.
In particular, we compute the proximal mapping of the rank term $\Gamma(\P)$ approximately in linear time via a fast \textit{orthogonal matching pursuit} (OMP) algorithm \cite{Tropp2004} that sparsely selects rank-1 matrices in dictionary $\cD$---corresponding to eigenvectors of only the largest eigenvalues, avoiding the need for computationally-expensive full eigen-decomposition. %where the few rank-1 update matrices represent interactions between important nodes.

We unroll iterations of our algorithm into layers to build a lightweight neural net for limited data-driven parameter tuning.
Experiments show that our proposal outperforms existing time-varying graph models, such as \cite{yamada19} and \cite{kalofolias17}, in signal interpolation for both synthetic and real-world datasets. 

% In experiments on two real-world datasets, our joint optimization produced better signal interpolation than existing graph estimation schemes \cite{sarwar02,kalofolias17,yamada19}. 

Compared to our conference version \cite{bagheri24}, we make several notable improvements:
\begin{itemize}
\item Instead of convexifying the matrix rank $\Gamma(\W^{(2)} - \W^{(1)})$ to the nuclear norm $\|\W^{(2)} - \W^{(1)}\|_*$, whic has complexity $\cO(N^3)$ to compute, we minimize the rank term directly via an OMP algorithm that sparsely selects rank-1 matrices from the dictionary $\cD$
%---composed of outer-product of $\W^{(1)}$'s eigenvectors---
to construct $\P = \W^{(2)} - \W^{(1)}$. 
A fast variant of OMP executes in linear time.
%\item We further propose a fast greedy approximation algorithm that computes only $B$ eigenvectors of $\W^{(1)}$ corresponding to the $B$ largest eigenvalues and reduces complexity to $\cO(N)$, for $B \ll N$. 
\item To justify our approximation of the low-rank update matrix $\P = \W^{(2)} - \W^{(1)}$ using a sparse combination of outer-products of $\W^{(1)}$'s eigenvectors, we provide a mathematical proof based on the CoMP \cite{Ledoux2001}.
%\textit{Concentration of Measure Phenomenon} (CoMP) \cite{Ledoux2001}.
The approximation defaults to the outer-product $\v_1 \v_1^\top$ of the Perron vector $\v_1$ that captures interactions among important nodes.
\item We incorporate a diagonal matrix $\Q$ into the fidelity term,  
%as $(\y_2 - \H^{(2)} \x_2)^\top \Q (\y_2 - \H^{(2)} \x_2)$, 
where the diagonal entries $Q_{i,i}$'s specify reliabilities of individual observed samples.
\item We unroll iterations of our algorithm into layers of a feed-forward network, enabling end-to-end learning of prior weights and reliability weights $Q_{i,i}$'s, as well as a shallow \textit{graph convolutional network} (GCN) to initialize the adjacency matrix $\W^{(1)}$.
\end{itemize}

The paper outline is as follows.
We first provide an overview of related work in \autoref{sec:related}.
We formulate our joint graph learning/signal interpolation problem in \autoref{sec:formulate}.
We develop our alternating algorithm in \autoref{sec:algo} and unroll it into a neural net in \autoref{sec:unrolling}.
Experiments and conclusion are presented in \autoref{sec:results} 
and \autoref{sec:conclude}, respectively. 

\section{Related Work}
\label{sec:related}
\subsection{Time-Varying Graph Learning}

Real-world scenarios often involve evolving node-to-node similarities \cite{prasad06,Cai17}, necessitating that similarity graphs also vary over time. 
Static graphs fail to reflect these changes, potentially leading to inaccurate representations of pairwise relationships. 
To address this, existing approaches for parameterizing slowly time-varying graphs fall into three categories.

The first category involves statistical approaches, 
where inverse covariance (precision) matrices---interpreted as graph Laplacians---describe the changing second-order statistics over time. 
Covariance matrices are re-estimated at different time instances under the assumption of “smooth” changes in their entries \cite{natali22,Ying23}.
However, these methods assume that there is sufficient data to reliably re-estimate covariance matrices, a condition that may not be met in data-starved scenarios.

The second category makes structural assumptions in the \textit{nodal} domain, such as smooth temporal variations or sparsity in the graph topology \cite{kalofolias17,yamada19}. For example, \cite{kalofolias17} assumes small Frobenius-norm changes in the adjacency matrix, while \cite{yamada19} imposes an $\ell_1$-regularization term to ensure that only a few edges adjust over time.

The third category focuses on the \textit{spectral} domain, assuming that the first $k$ singular vectors of the adjacency or Laplacian matrix remain stable while re-estimating the remaining $N-k$ \cite{chentong16,Li17,sarwar02}. This strategy preserves a core low-rank structure and updates finer details as the graph evolves.

Beyond these parametric methods, \textit{graph neural networks} (GNNs) provide a data-driven framework for modeling time-varying graphs. 
By combining temporal modeling (\eg, recurrent units) with graph convolutions, GNNs capture both evolving spatial and temporal patterns \cite{yu18,Zhao20,zheng19}. 
Dynamic GNNs \cite{manessi20,li19,goyal18} and graph lifelong learning methods \cite{chen22,zhou21,shenglisun19} adapt to changing topologies over time. 
However, their shortcomings include a lack of interpretability for the solutions obtained and the need for substantial training datasets to tune large numbers of parameters, making them unsuitable for data-starved scenarios.

%\vspace{-0.1in}
\subsection{Graph Signal Restoration}

Interpolation is a type of signal restoration, and graph-based signal restoration methods leverage graph-structured data kernels to facilitate recovery.
A common graph smoothness prior is the \textit{graph Laplacian regularizer} (GLR) \cite{pang17}, which quantifies signal smoothness over a defined graph kernel, and has been used for denoising \cite{pang17}, dequantization \cite{liu17}, and interpolation \cite{chen24}. We employ GLR in our interpolation framework.

%Subsequently, the \textit{Left Eigenvectors of the Random Walk Graph Laplacian} (LeRAG) \cite{liu17} introduces an alternative smoothness prior. It utilizes the left eigenvectors of the random walk graph Laplacian $\L_r \triangleq \D^{-1}\L$, which provides similar low-pass filtering properties as GLR but with normalized frequencies. This method was used for soft decoding of JPEG-compressed images.

\textit{Graph total variation} (GTV) \cite{couprie13,berger17} minimizes the weighted absolute differences of samples between connected nodes---an $\ell_1$-norm counterpart to GLR. 
It was shown that the signal-dependent GTV converges faster to piecewise-constant (PWC) signal reconstruction and was successfully employed in image deblurring \cite{bai19}. 
We employ GLR instead for signal interpolation due to its convenient differential quadratic form.

%While minimizing signal-dependent GLR / GTV promotes PWC signal reconstruction, it can suffer from the well-known “staircase” effect in slowly-varying spatial regions. 
%To alleviate this effect, a higher-order graph smoothness prior called the \textit{gradient graph Laplacian regularizer} (GGLR) \cite{chen24} has been proposed for image restoration problems. 
%By capturing local horizontal and vertical gradients using structure tensors \cite{knutsson11}, 
%GGLR mitigates noise-related challenges and reduces staircase artifacts. 
%GGLR promotes piecewise planar (PWP) reconstructions that yield higher-fidelity signal recovery compared to GLR.

\vspace{-0.1in}
\subsection{Graph Algorithm Unrolling}

To the best of our knowledge, the first deep unrolling of a graph-based algorithm is \textit{deep graph Laplacian regularizer} (DGLR) for image denoising \cite{zeng19}.
Subsequent works used \textit{deep graph total variation} (DGTV) for natural image denoising \cite{vu21} and light field image denoising \cite{yoshida22}.
In \cite{nakahama22}, the authors unrolled an ADMM-based iterative algorithm that minimizes an objective regularized by both GTV and GLR, but the underlying graph is statically fixed \textit{a priori} and not learned from data. 
Recently, \cite{Do2024} showed that a graph learning module with edge weight normalization, whose parameters are tuned using data, is akin to the self-attention mechanism in a transformer, and thus unrolling a graph-based algorithm results in a lightweight transformer-like neural net. 
The learning module for the initial graph in our work is inspired by \cite{Do2024}, but because of our assumed data-starved scenario, we rely on our proposed low-rank prior in updating each adjacency matrix $\W^{(t)}$ at instant $t$ to $\W^{(t+1)}$ at next instant $t+1$.

% \section{Preliminaries}
% \label{sec:prelim}
% \input{prelim.tex}

\section{Problem Formulation}
\label{sec:formulate}
%\vspace{-0.05in}
\subsection{GSP Definitions}

A graph $\cG(\cN,\cE,\W)$ contains a node set $\cN = \{1, \ldots, N\}$ of $N$ nodes and an edge set $\cE$, where edge $(i,j) \in \cE$ has weight $W_{i,j}$ specified by the $(i,j)$-th entry of the \textit{adjacency matrix} $\W$. 
If $(i,j) \not\in \cE$, then $W_{i,j} = 0$. 
We assume that $\cG$ has no self-loops, and thus $W_{i,i} = 0, \forall i$.
We assume an undirected graph, and thus $W_{i,j} = W_{j,i}, \forall i,j$, and $\W$ is symmetric.
A diagonal \textit{degree matrix} $\D \in \mathbb{R}^{N \times N}$ has diagonal entries $D_{i,i} = \sum_{j} W_{i,j}$.
Finally, a \textit{combinatorial graph Laplacian matrix} is defined as $\L \triangleq \D - \W$ \cite{ortega18ieee}. 
For a \textit{positive} graph, where $w_{i,j} \geq 0, \forall i,j$, it can be proven \cite{cheung18} that $\L$ is \textit{positive semi-definite} (PSD), \ie, $\x^\top \L \x \geq 0, \forall \x \in \mathbb{R}^N$, or equivalently, all eigenvalues $\{\lambda_k\}$ of $\L$ are non-negative. 

\subsection{Graph Laplacian Regularizer}
%\paragraph*{Graph Laplacian Regularizer}

The extent to which a signal $\x \in \mathbb{R}^N$ is smooth w.r.t. a graph $\cG$ specified by Laplacian $\L$ can be quantified using the \textit{graph Laplacian regularizer} (GLR) \cite{pang17}:
\begin{align}
\x^\top \L \x = \sum_{(i,j) \in \cE} w_{i,j} (x_i - x_j)^2 .   
\end{align}
The GLR can be used as a signal prior for graph signal restoration tasks such as denoising \cite{pang17}, dequantization \cite{liu17}, and interpolation \cite{chen24}. 
Specifically for interpolation, denote by $\H \in \{0,1\}^{M \times N}$ a sampling matrix that picks out $M$ observable samples from $N$, where $M < N$.
The interpolation problem can be formulated as
\begin{align}
\min_\x \|\y - \H \x\|^2_2 + \mu \, \x^\top \L \x,
\label{eq:interpolate}
\end{align}
where $\mu \in \mathbb{R}_+$ is a weight parameter trading off the fidelity term and prior. 
Given that $\L$ is PSD, the solution $\x^*$ to \eqref{eq:interpolate} can be computed by solving the following linear system:
\begin{equation}
(\H^\top \H + \mu \L) \x^* = \H^\top \y .   
\label{eq:linSys}
\end{equation}
If $\L$ is sparse, then \eqref{eq:linSys} can be solved in linear time using \textit{conjugate gradient} (CG) \cite{shewchuk94}, given that the coefficient matrix $\H^\top \H + \mu \L$ is symmetric, PD, and sparse.

\subsection{Time-varying Graph Signal Interpolation}
\label{subsec:interpolationformula}

We consider an interpolation scenario, where signal $\x_1$ and the underlying graph $\cG_1$---specified by a (sparse\footnote{Graph models used in the GSP literature are typically sparse for computation reasons \cite{ortega18ieee,cheung18}. 
We will assume the initial $\W^{(1)}$ is sparse, and ensure that the subsequent constructed $\W^{(t)}$'s for $t>1$ are also sparse in our algorithm design.}) adjacency matrix $\W^{(1)}$---at time instant $t=1$ are known.
This is practical for two possible scenarios. 
In the first scenario, \textit{full} signal observations from instants $t\leq 1$ are available by $t=1$.
In the second scenario, at time $t=T$, we first collect available \textit{partial} observations at $t \in \{1, \ldots, T\}$ and construct a rough estimate $\x_1$ at $t=1$, and then we interpolate one signal at a time starting at $t=2$.
In either case, a single empirical covariance matrix can be computed from the collected data, and in turn, a sparse inverse covariance matrix (interpreted as graph Laplacian) can be estimated using techniques such as GLASSO and CLIME \cite{friedman08,caitony11}, leading to $\W^{(1)}$.
An alternative is to compute initial $\W^{(1)}$ in a data-driven manner via a shallow GCN, as described in Section\;\ref{subsec:initialGraph}.

Then, at instants $t \in \{2, \ldots, T\}$, we assume that partial signal observations $\y_t \in \mathbb{R}^{M_t}$, $M_t < N$, are available\footnote{$\y_t$ can also be viewed as a set of samples observed between discrete time instants $t-1$ and $t$.}.
The goal is to interpolate $\y_t$ to $\x_t \in \mathbb{R}^N$ for each instant $t$, assuming that $\x_t$ is smooth with respect to (w.r.t.) to graph $\cG_t$, and that the adjacency matrices $\W^{(t)}$'s are slowly time-varying. 
Denote by $\X \in \mathbb{R}^{N \times T}$ a matrix variable that contains $\x_t$ as the $t$-th column.  
Our problem can be interpreted as \textit{matrix completion}: given sparse entries $\y_t$ in each column $t$, complete the remaining $N-M_t$ entries to restore the entire $N$-by-$T$ data matrix $\X$.

A matrix can be completed assuming a low-rank prior using methods such as \textit{robust principal component analysis} (RPCA) \cite{candes11}. 
In contrast, by assuming that each signal $\x_t$ at time $t$ is smooth w.r.t. underlying graph $\cG_t$ that is slowly time-varying, we model spatial similarities of each $\x_t$ more specifically, resulting in better signal interpolation.
See Section\;\ref{sec:results} for empirical comparison among competing schemes. 

\vspace{0.05in}
\noindent
\textbf{Concrete Example}:
One real-world example of this setting is \textit{farmland valuation}: 
Given that crop-growing fields are sold only in some counties of a farming state like Iowa in a given quarter, the problem is to estimate field prices per acre in other counties with no recorded sales. 
Pairwise similarities between neighboring counties (with similar environmental variables such as rainfall and soil composition) can change slowly over time due to regional factors, such as migrating population, county-specific tax policies, and farming practices.

\subsection{Slowly Time-Varying Graph Model}

In our slowly time-varying graph model, 
we assume that the difference between consecutive symmetric adjacency matrices $\W^{(t+1)}$ and $\W^{(t)}$, corresponding to graphs $\cG_{t+1}$ and $\cG_{t}$  at times $t+1$ and $t$, can be approximated as  a low-rank matrix of rank $K$, $K \ll N$, \ie,
\begin{align}
\W^{(t+1)} = \W^{(t)} + \sum_{k=1}^K \alpha_k^{(t)} \u_k^{(t)} (\u_k^{(t)})^\top,
\label{eq:lowRank}
\end{align}
where each pair of scalar $\alpha_k^{(t)} \in \mathbb{R}$ and unit-norm vector $\u_k^{(t)} \in \mathbb{R}^N$ specify a rank-1 update matrix.

\vspace{0.05in}
\noindent
\textbf{Concrete Example:}
Continuing our farmland valuation example, consider the case where a county $i$'s average field price increases at time $t+1$ thanks to a new municipal tax exemption. 
Consequently, county $i$'s pairwise similarities with \textit{all} other counties change----a \textit{rank-$2$ update} in $\W^{(t+1)}$ due to changes in row $i$ and column $i$.
Note that this rank-$2$ update may not be small in Frobenius norm $\|\W^{(t+1)} - \W^{(t)}\|^2_F$ \cite{kalofolias17}, or small in $\ell_0$-norm $\|\W^{(t+1)} - \W^{(t)}\|_0$  \cite{yamada19}), as previous models assumed. 
We validate this low-rank assumption \eqref{eq:lowRank} for small time intervals $[t,t+1]$ empirically in Section\;\ref{sec:results} using real-world datasets.

%\vspace{0.05in}
With our time-varying graph model \eqref{eq:lowRank}, signal $\x_t$ can be characterized statistically given adjacency matrix $\W^{(t)}$ and signal $\x_{t-1}$ as
%\vspace{-0.05in}
%\begin{small}
\begin{align}
\text{Pr}(\x_t|\x_{t-1}) = \cN(\x_t | \0, (\L^{(t)} + e \I)^{-1}) \, \cN(\x_t|\x_{t-1}, \beta \I) ,
\label{eq:signalModel}
\end{align}
%\end{small}\noindent 
where $e \in \mathbb{R}_+$ is a small positive parameter to ensure invertibility, given that the combinatorial graph Laplacian $\L^{(t)}$ has an eigenvalue zero. 
\eqref{eq:signalModel} states that probability $\text{Pr}(\x_t|\x_{t-1})$ is proportional to \textit{both} the Gaussian-distributed probability of $\x_t$ with zero mean and covariance $(\L^{(t)} + \epsilon \I)^{-1}$ \textit{and} the iid Gaussian-distributed probability of $\x_t$ with mean at $\x_{t-1}$ and variance $\beta$.

\subsection{Problem Formulation}

Consider the isolated case when we seek to reconstruct signal $\x_2 \in \mathbb{R}^N$ at time $t=2$, given signal $\x_1 \in \mathbb{R}^N$ and adjacency matrix $\W^{(1)} \in \mathbb{R}^{N \times N}$ at time $t=1$, as well as partial observation $\y_2 \in \mathbb{R}^{M_2}$ at time $t=2$, where $M_2 < N$. 
Denote by $\Q \in \mathbb{R}^{M_2 \times M_2}$ a diagonal matrix where $Q_{i,i}$ specifies the reliability of observation $y_i$ (to be discussed in detail later).
Denote by $\H^{(2)} \in \{0,1\}^{M_2 \times N}$ a sampling matrix that picks out $M_2$ observed entries from $\x_2$, \ie,
\begin{align}
H_{i,j} = \left\{ \begin{array}{ll}
1 & \mbox{if the $i$-th sample is index $j$} \\
0 & \mbox{o.w.}
\end{array} \right. .
\end{align}
We formulate a joint optimization to obtain the adjacency matrix $\W^{(2)}$ \textit{and} the signal $\x_2$. 
Given signal model \eqref{eq:signalModel}, for $\x_2$ we employ GLR\footnote{We employ the  GLR $\x^\top \L \x$ as a smoothness prior because its $\ell_2$-norm form is amenable to fast optimization. Other graph signal priors, such as graph total variation (GTV) \cite{bai19}, can also be used.} $\x_2^\top \L^{(2)} \x_2$ \cite{pang17} and $\ell_2$-norm $\|\x_2 - \x_1\|^2_2$ to promote graph smoothness and similarity to $\x_1$.
Given time-varying graph model \eqref{eq:lowRank}, for $\W^{(2)}$ we use a low-rank prior $\Gamma(\W^{(2)} - \W^{(1)})$ to promote similarities between $\W^{(2)}$ and $\W^{(1)}$.
This results in the following optimization formulation:
\begin{align}
\min_{\W^{(2)}, \x_2} & (\y_2 - \H^{(2)} \x_2)^\top \Q (\y_2 - \H^{(2)} \x_2) + \mu \, \x_2^\top \L^{(2)} \x_2   \nonumber \\
&+ \xi \|\x_2 - \x_1\|^2_2 + \eta \, \Gamma(\W^{(2)} - \W^{(1)}),  \nonumber \\ 
\mbox{s.t.} & ~~ \L^{(2)} = \mathrm{diag}(\W^{(2)} \1) - \W^{(2)} 
\nonumber \\
& ~~ W^{(2)}_{i,i} = 0, ~\forall i, ~~~ W^{(2)}_{j,i} = W^{(2)}_{i,j} \geq 0, ~\forall i \neq j
\label{eq:jointForm}
\end{align}
where $\mu, \xi, \eta \in \mathbb{R}_+$ are weight parameters, and $\Gamma(\M)$ is the rank of matrix $\M$. 
%The first constraint defines matrix $\M$ that is a low-rank update from $\W^{(1)}$, which has a small Frobenius norm from $\W^{(2)}$, according to the low-rank model in \eqref{eq:lowRank}. 
The first constraint ensures a properly defined combinatorial Laplacian matrix $\L^{(2)}$ from $\W^{(2)}$.
The last two constraints ensure that $\W^{(2)}$ is an adjacency matrix for a positive undirected graph $\cG_2$ without self-loops. 
%The fifth constraint ensures that the node degree total for $\cG^{t+1}$ remains no smaller than $\cG^t$ (and thus prevents the pathological solution when all nodes become disconnected and GLR goes to zero).

\section{Algorithm Development}
\label{sec:algo}
\subsection{Alternating Optimization Approach}
\label{Alternating Optimization Algorithm}

We solve \eqref{eq:jointForm} via an alternating approach: solve for one variable while holding the other fixed, and vice versa, until solution convergence.  
When $\W^{(2)}$ is fixed, \eqref{eq:jointForm} becomes an unconstrained convex \textit{quadratic programming} (QP) problem, whose solution can be obtained by solving a system of linear equations similar to \eqref{eq:linSys}:

\vspace{-0.05in}
\begin{small}
\begin{align}
\left( (\H^{(2)})^\top \Q \H^{(2)} + \mu \L^{(2)} + \xi \I \right) \x_2^* &= (\H^{(2)})^\top \Q \y_2 + \xi \x_1 .
\label{eq:linSys1}
\end{align}
\end{small}\noindent
Assuming that $\L^{(2)}$ is sparse (graph model $\cG_t$ is assumed sparse), the coefficient matrix $(\H^{(2)})^\top \Q \H^{(2)} + \mu \L^{(2)} + \xi \I$ is symmetric, PD, and sparse, and we can compute optimal $\x_2^*$ via CG in linear time, as described previously. 

When $\x_2$ is fixed, we rewrite the optimization \eqref{eq:jointForm} as follows.
We first define an \textit{indicator function} $\text{I}_W(\W)$ for matrix $\W$ as

\vspace{-0.1in}
\begin{small}
\begin{align}
\text{I}_W(\W) &= \left\{
\begin{array}{ll}
0, & \mbox{if}~~ W_{i,i} = 0, \forall i, ~
W_{j,i} = W_{i,j} \geq 0, \forall i \neq j
\\
\infty, & \mbox{o.w.}
\end{array} 
\right. .
\label{eq:ind}
\end{align}
\end{small}\noindent
Denote by $\cC$ the set of matrices such that $\text{I}_W(\W) = 0$, \ie, $\cC = \{\W \in \mathbb{R}^{N \times N} \,|\, \text{I}_W(\W) = 0 \}$.
We prove that $\cC$ is convex. 
\begin{lemma}
$\cC$ is a convex set.
\label{lemma:convexSet}
\end{lemma}
\begin{proof}
Denote by $\W_1, \W_2 \in \cC$ two matrices in set $\cC$.  
Consider $\W = \rho \W_1 + (1-\rho) \W_2$ for $0 \leq \rho \leq 1$, a convex combination of $\W_1$ and $\W_2$. 
Diagonal entry $i$ of $\W$ is
\begin{align}
W_{i,i} &= \rho W_{1,i,i} + (1-\rho) W_{2,i,i} 
\nonumber \\
&= \rho 0 + (1-\rho) 0 = 0 .
\end{align}
On the other hand, off-diagonal entry $W_{i,j}$ is
\begin{align}
W_{i,j} &= \rho W_{1,i,j} + (1-\rho) W_{2,i,j} \geq 0
\\
W_{j,i} &= \rho W_{1,j,i} + (1-\rho) W_{2,j,i} 
\nonumber \\
&= \rho W_{1,i,j} + (1-\rho) W_{2,i,j} = W_{i,j} .
\end{align}
Thus, $\text{I}_W(\W) = 0$, and $\cC$ is a convex set.
\end{proof}

The implication of Lemma\;\ref{lemma:convexSet} is that $\text{I}_W(\W)$ is convex, albeit non-smooth.  
To rewrite the objective in \eqref{eq:jointForm}, we first rewrite GLR as
\begin{align}
\x_2^\top \L^{(2)} \x_2 &= 
\text{Tr}(\x_2^\top \L^{(2)} \x_2) = 
\text{Tr}(\L^{(2)} \x_2 \x_2^\top) .
\end{align}
Expanding out the definition of $\L^{(2)}$ as variable $\W^{(2)}$, we can now define an unconstrained optimization for $\W^{(2)}$ as
\begin{align}
\min_{\W^{(2)}} & ~~\mu \, \text{Tr}((\mathrm{diag}(\W^{(2)} \1) - \W^{(2)}) \x_2 \x_2^\top)  
\nonumber \\
& ~~ + \eta \, \Gamma(\W^{(2)} - \W^{(1)})
+ \text{I}_W(\W^{(2)}) .
\label{eq:uOpt}
\end{align}
The objective in \eqref{eq:uOpt} is composed of three terms: the first is convex and differentiable (smooth), the second is combinatorial,  and the third is convex but non-smooth.

\subsection{Proximal Gradient Descent}

We minimize the three terms in \eqref{eq:uOpt} one at a time via a variant\footnote{Because matrix rank $\Gamma(\cdot)$ is evaluated on the \textit{difference} of matrices $\W^{(2)} - \W^{(1)}$, not the matrix variable $\W^{(2)}$, the proximal mapping for $\W^{(2)} - \W^{(1)}$ is unique, and we cannot convexify it to a nuclear norm for simple soft-thresholding operations, as done in \cite{parikh13}.} of \textit{proximal gradient descent} (PGD) \cite{parikh13} as follows.

\vspace{0.1in}
\subsubsection{Gradient Descent}
\label{subsubsection:gradientDescent}

The first term can be expanded as
\begin{align}
\label{eq: firstterm}
\mu \, \text{Tr}(\mathrm{diag}(\W^{(2)} \1) (\x_2) (\x_2)^\top) -
\mu \; \text{Tr}(\W^{(2)} (\x_2) (\x_2)^\top) .
\end{align}
Given that $\mathrm{diag}(\a) = (\a \1^\top) \odot \I$ where $\a \in \mathbb{R}^{N}$, we can rewrite the first part as $ \mu \, \text{Tr}(((\W^{(2)} \1 \1^\top) \odot \I) \x_2 \x_2^\top)$. 
Thus, the gradient $\nabla_{\W^{(2)}}$  of \eqref{eq: firstterm} w.r.t. $\W^{(2)}$ is 
\begin{align}
\nabla_{\W^{(2)}} = \mu \, (\x_2 \odot \x_2) \1^\top  - \mu \, \x_2 \x_2^\top .
\end{align}
Thus, to reduce the first term, we compute $\W^{(2)} - \epsilon \nabla_{\W^{(2)}}$, where $\epsilon > 0$ is a step size.

While the gradient $\nabla_{\W^{(2)}}$ for the first term in \eqref{eq:uOpt} is easily computable, reducing the second rank term via proximal mapping is more challenging.

\vspace{0.1in}
\subsubsection{Proximal Mapping for Matrix Rank Term}

By definition, the \textit{proximal mapping} \cite{parikh13} for the weighted rank function $\eta g(\W^{(2)}) = \eta \, \Gamma(\W^{(2)} - \W^{(1)})$ is

\vspace{-0.05in}
\begin{small}
\begin{align}
\text{prox}_{\eta g}(\M) &= \arg\min_{\Z} \frac{1}{2} \|\Z - \M\|^2_F + \eta \, \Gamma(\Z - \W^{(1)}) 
\label{eq:proxRank}
\end{align}
\end{small}\noindent
where $\M$ is a candidate solution for $\W^{(2)}$.
This formulation encourages $\Z$ to remain close to $\M$ in Frobenius norm while penalizing the rank of its deviation from $\W^{(1)}$.

The proximal mapping \eqref{eq:proxRank} is itself an optimization, and we efficiently approximate it as follows.
%\footnote{In the conference version \cite{bagheri24}, we first convexify rank $\Gamma(\W^{(2)} - \W^{(1)})$ to the nuclear norm---the sum of singular values---with complexity $\cO(N^3)$. In this work, we design a fast OMP algorithm addressing the rank term directly with complexity $\cO(N)$.}.
Note first that $\W^{(1)}$ is an adjacency matrix for an undirected graph, and hence $\W^{(1)}$ is real and symmetric, thus eigen-decomposable by the Spectral Theorem \cite{HAWKINS75}. 
Denote by $\V \in \mathbb{R}^{N \times N}$ a matrix containing the orthonormal eigenvectors of $\W^{(1)}$ as columns, \ie, $\W^{(1)} = \V \bLambda \V^\top$, where $\bLambda \in \mathbb{R}^{N \times N}$ is a diagonal matrix $\bLambda = \mathrm{diag}(\lambda_1, \lambda_2, \ldots, \lambda_N)$, and $|\lambda_1| \geq |\lambda_2| \geq \cdots \geq |\lambda_N|$. 

%Second, given the assumed slowly time-varying nature of adjacency matrices \eqref{eq:lowRank}, we assume that the eigenvectors corresponding to non-zero eigenvalues of (possibly low-rank) difference matrix $\W^{(1)} - \M$ are similar to eigenvectors of $\W^{(1)}$. 

\vspace{0.05in}
\subsubsection{Full OMP Algorithm}

We first develop a full greedy algorithm for \eqref{eq:proxRank} mimicking the \textit{orthogonal matching pursuit} (OMP) algorithm \cite{Tropp2004}, albeit with high complexity; a low-complexity variant is presented in the sequel.
Note first that by spectral decomposition, we can write $\W^{(1)} = \sum_{i=1}^N \lambda_i \v_i \v_i^\top$ as a linear combination of rank-1 matrices (outer-product of orthonormal eigenvectors $\v_i$'s), each scaled by eigenvalue $\lambda_i$.
To achieve minimal matrix rank $\Gamma(\Z - \W^{(1)})$, we first initialize $\Z \leftarrow \W^{(1)}$, which means 
\begin{align}
\Gamma (\Z - \W^{(1)}) = \Gamma(\W^{(1)} - \W^{(1)}) = 0. 
\end{align}

Next, we project $\Z - \M$ onto rank-1 matrix $\v_1 \v_1^\top$ by computing a matrix inner-product:
\begin{align}
a_1 = \langle \Z - \M, \v_1 \v_1^\top \rangle \triangleq \text{Tr}((\Z-\M)\v_1 \v_1^\top) .
\end{align}
To maximally reduce Frobenius norm $\|\Z - \M\|_F^2$ using matrix component $\v_1 \v_1^\top$, we update variable $\Z$ to $\Z = (\lambda_1 - a_1) \v_1 \v_1^\top + \sum_{i=2}^N \lambda_i \v_i \v_i^\top$.
Consequently, rank of $\Z-\W^{(1)}$ is incremented by 1, \ie, $\Gamma(\Z - \W^{(1)}) = 1$.

To further reduce Frobenius norm $\|\Z - \M\|_F^2$ by gradually incrementing rank $\Gamma(\Z-\W^{(1)})$, we first define a \textit{dictionary} $\cD$ containing $3N-3$ rank-1 symmetric matrices: 
\begin{align}
\cD \triangleq \left\{ ~ \{\v_i \v_i^\top\}_{i=2}^N, ~\{\g_i \g_i^\top\}_{i=2}^N, ~\{\h_i \h_i^\top \}_{i=2}^N ~ \right\}
\label{eq:setR}
\end{align}
where $\g_i$ and $\h_i$, $i \in \{2, \ldots, N\}$, are defined as
\begin{align}
\g_i \triangleq \frac{1}{\sqrt{2}} \left( \v_1 + \v_i \right),
~~~
\h_i \triangleq \frac{1}{\sqrt{2}} \left( \v_1 - \v_i \right).
\label{eq:setR2}
\end{align}
At each iteration, we project $\Z - \M$ onto each \textit{triple} $i$ of three rank-1 matrices $(\v_i \v_i^\top, \g_i \g_i^\top, \h_i \h_i^\top)$ in $\cD$: 
\begin{align}
a_i &= \langle \Z - \M, \v_i \v_i^\top \rangle, ~~~ \forall i \in \{2, \ldots, N\}
\nonumber \\
b_i &= \langle \Z - \M - a_i \v_i \v_i^\top, \g_i \g_i^\top \rangle
\nonumber \\
c_i &= \langle \Z - \M - a_i \v_i \v_i^\top - b_i \g_i \g_i^\top, \h_i \h_i^\top \rangle .
\end{align} 
The index $i$ with the largest magnitude sum $|a_i|+|b_i|+|c_i|$ among $N-1$ triples constitutes approximation $a_i \v_i \v_i^\top + b_i \g_i \g_i^\top + c_i \h_i \h_i^\top$ of error $\Z-\M$ given dictionary $\cD$ that increases rank $\Gamma(\Z - \W^{(1)})$ by one, since $\v_i, \g_i, \h_i$ introduces a new direction $\v_i$ into the set of vectors that compose the rank-1 matrices constructing $\Z - \W^{(1)}$.

Specifically, denote by $\cI$ the index set that composes current \textit{correction matrix} $\R = \Z - \W^{(1)}$, \ie,
\begin{align}
\R = a_1 \v_1 \v_1^\top + \sum_{i \in \cI} \left( a_i \v_i \v_i^\top + b_i \g_i \g_i^\top + c_i \h_i \h_i^\top \right) .
\end{align}
To ensure that residual $\M - \R$ is orthogonal to the selected rank-1 components, we solve the following least-square problem to update coefficients $a_1$ and $a_i, b_i, c_i, \forall i \in \cI$:

\vspace{-0.1in}
\begin{scriptsize}
\begin{align}
\min_{a_1, a_i, b_i, c_i \in \cI} \left\| 
\W^{(1)} + a_1 \v_1 \v_1^\top + \sum_{i \in \cI} \left( a_i \v_i \v_i^\top + b_i \g_i \g_i^\top + c_i \h_i \h_i^\top \right) - \M
\right\|_F^2 
\label{eq:LS}
\end{align}
\end{scriptsize}\noindent
The optimal coefficients $\a$ can be obtained by solving a linear system:
\begin{align}
\V^\top \V \a = \V^\top \m
\label{eq:LS_soln}
\end{align}
where
\begin{small}
\begin{align}
\a = \left[ \begin{array}{c}
a_1 \\
a_i \\
\vdots \\
b_i \\
\vdots \\
c_i \\
\vdots 
\end{array} \right], ~
\V^\top = \left[ \begin{array}{c}
\text{vec}^\top(\v_1 \v_1^\top) \\
\text{vec}^\top(\v_i \v_i^\top) \\
\vdots \\
\text{vec}^\top(\g_i \g_i^\top) \\
\vdots \\
\text{vec}^\top(\h_i \h_i^\top) \\
\vdots 
\end{array} \right], ~
\m = \text{vec}(\W^{(1)} - \M)  .
\label{eq:LS_defn}
\end{align}
\end{small}\noindent
$\text{vec}(\M)$ vectorizes $\M$ into a column vector $\m$ in row-major order. 
See Appendix\;\ref{append:LS_soln} for a derivation.

We then check if objective \eqref{eq:proxRank} has decreased. 
If so, we continue to select the next index $i$ with the largest magnitude sum $|a_i|+|b_i|+|c_i|$ into the support. 
We continue till objective \eqref{eq:proxRank} increases, at which point we retain the previous solution $\Z$.
See Fig.\;\ref{fig:greedy_algo_flowchart} for a flow chart of the algorithm.

\begin{figure}[ht]
    \centering
    \includegraphics[width=0.5\textwidth]{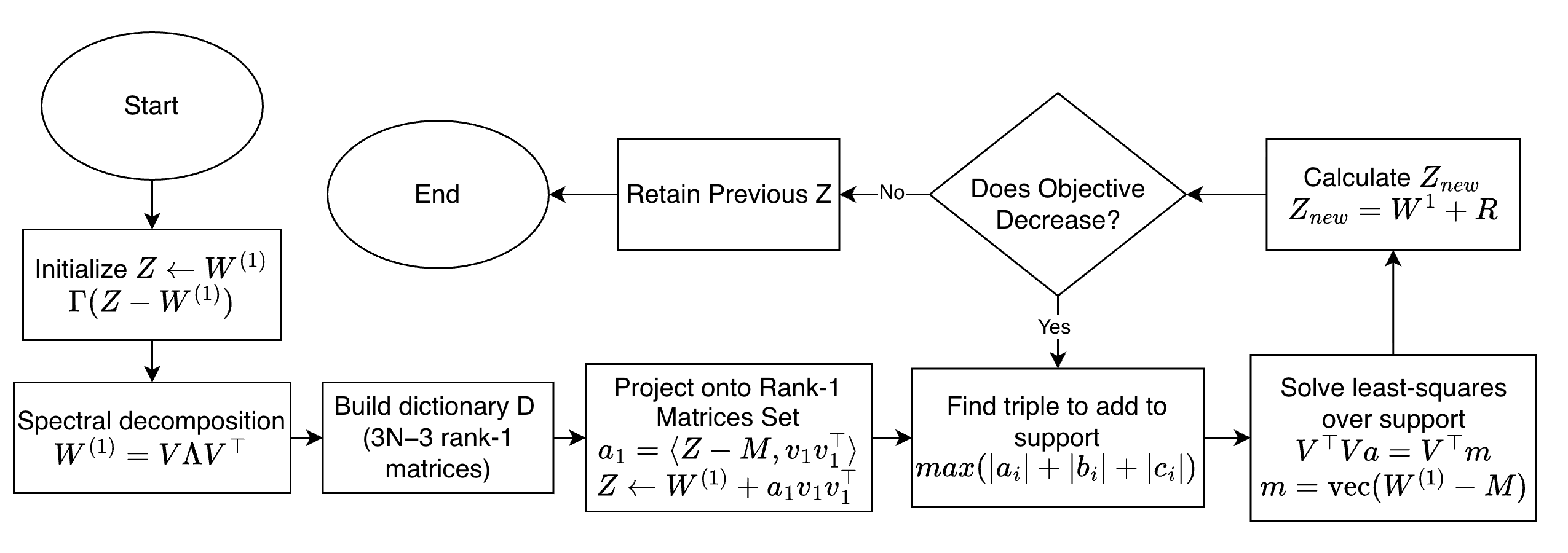}
    \vspace{-0.25in}
    \caption{Flow chart for the full OMP algorithm. }
    \label{fig:greedy_algo_flowchart}
\end{figure}

Next assumption is necessary for the full OMP algorithm to compute an optimal solution to \eqref{eq:proxRank} for a particular $\eta$:

\vspace{0.05in}
\noindent
\textbf{Assumption}: Matrix difference $\W^{(1)} - \M$ is $K$-sparse in the span of rank-1 matrix dictionary $\cD \cup \{\v_1 \v_1^\top\}$, \ie, $\W^{(1)} - \M = a_1 \v_1 \v_1^\top + \sum_{i \in \cI} (a_i \v_i \v_i^\top + b_i \g_i \g_i^\top + c_i \h_i \h_i^\top)$, where $|\cI| = K-1$, and $a_1 \neq 0$. 

\vspace{0.05in}
The assumption means that the slowly time-varying graph changes in consecutive adjacency matrices are well described by rank-1 matrices in $\cD \cup \{\v_1 \v_1^\top\}$.  
We defer the justification of this assumption to Section\;\ref{subsec:justify}. 
%We argue that the assumption is approximately true in practice.$\M$ is a candidate solution for $\W^{(2)}$, and because of the slowly time-varying assumption for consecutive adjacency matrices, $\W^{(2)}$ and $\W^{(1)}$ ought to be similar.

To show the algorithm's optimality, we first consider the following constrained variant of optimization \eqref{eq:proxRank}:
\begin{align}
\min_\Z \|\Z - \M \|^2_F,
~~\mbox{s.t.}~~ \Gamma(\Z - \W^{(1)}) \leq r .
\label{eq:rankConst}
\end{align}
The two optimizations \eqref{eq:proxRank} and \eqref{eq:rankConst} are related by the following two lemmas.

\begin{lemma}
Denote by $\Z^*$ an optimal solution to \eqref{eq:proxRank} for a given  $\eta^*$, such that $\Gamma(\Z^* - \W^{(1)}) = r^*$. 
$\Z^*$ is also an optimal solution to \eqref{eq:rankConst} for $r=r^*$.
\label{lemma:duality}
\end{lemma}

\begin{proof}
By definition of optimality, $\Z^*$ satisfies
\begin{align}
& \frac{1}{2} \|\Z^* - \M\|^2_F + \eta^* \Gamma (\Z^* - \W^{(1)}) \leq 
\\
&~~~~~~~~~~ \frac{1}{2} \|\Z - \M\|^2_F + \eta^* \Gamma (\Z - \W^{(1)}), ~~\forall \Z
\nonumber 
\end{align}
Define $\cS \triangleq \{\Z \,|\, \Gamma(\Z - \W^{(1)}) \leq r^*\}$. 
Then, 
\begin{align}
& \frac{1}{2} \|\Z^* - \M\|^2_F + \eta^* \Gamma (\Z^* - \W^{(1)}) \leq 
\\
&~~~~~~~~~~ \frac{1}{2} \|\Z - \M\|^2_F + \eta^* \Gamma (\Z - \W^{(1)}), ~~\forall \Z \in \cS 
\nonumber \\
& \frac{1}{2} \|\Z^* - \M\|^2_F + \eta^* \left( r^* - \Gamma (\Z - \W^{(1)}) \right) \leq 
\\
&~~~~~~~~~~ \frac{1}{2} \|\Z - \M\|^2_F, ~~ \forall \Z \in \cS
\nonumber \\
& \|\Z^* - \M\|^2_F \stackrel{(a)}{\leq} \|\Z - \M\|^2_F, ~~ \forall \Z \in \cS
\end{align}
where $(a)$ is true since $\eta^* > 0$, and $r^* \geq \Gamma(\Z - \W^{(1)})$ for $\Z \in \cS$. 
\end{proof}

In general, an optimal solution to either \eqref{eq:proxRank} or \eqref{eq:rankConst} may not be unique.
Suppose we compute an optimal solution $\Z^o$ to \eqref{eq:rankConst} for $r=r^*$, which may be different from a solution $\Z^*$ of \eqref{eq:proxRank}.
The following lemma relates $\Z^o$ to \eqref{eq:proxRank}.

\begin{lemma}
Denote by $\Z^*$ an optimal solution to \eqref{eq:proxRank} for a given $\eta^*$, such that $\Gamma(\Z^*-\W^{(1)}) = r^*$. Denote by $\Z^o$ an optimal solution to \eqref{eq:rankConst} for $r=r^*$.
$\Z^o$ is also an optimal solution to \eqref{eq:proxRank} for the same $\eta^*$.
\label{lemma:duality2}
\end{lemma}

\begin{proof}
By Lemma\;\ref{lemma:duality}, $\Z^*$ is also an optimal solution to \eqref{eq:rankConst} for $r=r^*$, and thus by optimality to \eqref{eq:rankConst}, $\|\Z^*-\M\|^2_F = \|\Z^o-\M\|^2_F$.
Further, $\Gamma(\Z^*-\W^{(1)}) = r^* \geq \Gamma(\Z^o - \W^{(1)})$. 
Thus,
\begin{align}
\|\Z^o-\M\|^2_F + \eta^* \Gamma(\Z^o-\W^{(1)}) \leq
\nonumber \\
\|\Z^*-\M\|^2_F + \eta^* \Gamma(\Z^*-\W^{(1)}) ,
\end{align}
given $\eta^* > 0$.

By optimality of $\Z^*$ to \eqref{eq:proxRank},
\begin{align}
\|\Z^*-\M\|^2_F + \eta^* \Gamma(\Z^*-\W^{(1)}) \leq
\nonumber \\
|\Z-\M\|^2_F + \eta^* \Gamma(\Z-\W^{(1)}), ~~ \forall \Z .
\end{align}
Thus, we can conclude
\begin{align}
\|\Z^o-\M\|^2_F + \eta^* \Gamma(\Z^o-\W^{(1)}) \leq
\nonumber \\
|\Z-\M\|^2_F + \eta^* \Gamma(\Z-\W^{(1)}), ~~ \forall \Z,
\end{align}
and $\Z^o$ is an optimal solution to \eqref{eq:proxRank}.
\end{proof}

We now state a crucial corollary of Lemma\;\ref{lemma:duality} and \ref{lemma:duality2}.

\begin{corollary}
The solution set to \eqref{eq:rankConst} for all possible $r$'s is a \textit{superset} that contains a set of solutions to \eqref{eq:proxRank} for all possible $\eta$'s for $\eta > 0$.
\label{corollary:superset}
\end{corollary}

\begin{proof}
By Lemma\;\ref{lemma:duality}, each solution $(\eta^*,\Z^*)$ with $\Gamma(\Z^* - \W^{(1)}) = r^*$, of the set of solutions $\{\Z^*\}$ to \eqref{eq:proxRank} for all possible $\eta > 0$, is an optimal solution to \eqref{eq:rankConst} for $r=r^*$.
By solving \eqref{eq:rankConst} instead for solution $\Z^o$ (which may not be the same as $\Z^*$), $\Z^o$ is also a solution to \eqref{eq:proxRank} at $\eta^*$.
Thus, solving \eqref{eq:rankConst} for all possible $r$'s creates a set of solutions to \eqref{eq:proxRank} for all possible $\eta$'s for $\eta > 0$.
\end{proof}

This corollary is useful in algorithm design, because while the set of possible $\eta > 0$ is infinite, the set of possible $r = \Gamma(\Z - \W^{(1)}) \in \{0, 1, \dots, N\}$ is countably finite.
Specifically, starting at $r=0$, we can generate solutions to \eqref{eq:rankConst} for increasing $r$ and test them against objective in \eqref{eq:proxRank} for a given $\eta$ for an optimal solution. 

For simplicity, we prove that a simple variant of full OMP---where an index $i$ is selected into the support set $\cI$ at each greedy step if $\r_i \in \{\v_i \v_i^\top, \g_i \g_i^\top, \h_i \h_i^\top\}$ induces the largest inner-product magnitude $|\langle \Z - \M, \r_i \rangle|$---computes an optimal solution to \eqref{eq:proxRank} under specified conditions.

\begin{theorem}
An optimal solution $\Z^o$ to \eqref{eq:proxRank} with $\Gamma(\Z^o - \W^{(1)}) = K$ for a given $\eta$, is computable by Full OMP, if $\W^{(1)} - \M$ is $K$-sparse in the span of rank-1 matrix set $\cD \cup \{\v_1 \v_1^\top\}$ for support set  $\cI \cup \{1\}$, \ie, $\W^{(1)} - \M = a_1 \v_1 \v_1^\top + \sum_{i \in \cI} a_i \v_i \v_i^\top + b_i \g_i \g_i^\top + c_i \h_i \h_i^\top$, where $a_1 \neq 0$ and $|\cI| = K-1$, and each set of coefficients $(a_m, b_m, c_m)$ for $m \in \cI$ satisfies one of two inequalities:

\vspace{-0.05in}
\begin{small}
\begin{align}
|a_m| &> \frac{|b_m|}{2} + \frac{|c_m|}{2} + \frac{1}{4} \sum_{i \in \cI}  \left( |b_i| + |c_i| \right)
\label{eq:opt_condition1} \\
\max \left(|b_m|,\, |c_m|\right) &> \frac{|a_m|}{2} - \frac{|b_m| + |c_m|}{4} + \frac{1}{2} \sum_{i \in \cI}  \left( |b_i| + |c_i| \right) .
\label{eq:opt_condition2}
\end{align}
\end{small}
\label{thm:optimality}
\end{theorem}

Essentially, \eqref{eq:opt_condition1} and \eqref{eq:opt_condition2} mean that the largest magnitude in each coefficient set $(a_m, b_m, c_m)$, $m \in \cI$, used to expressed $\W^{(1)} - \M$, must be relatively large, and $K$ must be small.
See Appendix\;\ref{append:optimality} for a proof.

Note that adding all three $\v_i \v_i^\top$, $\g_i \g_i^\top$ and $\h_i \h_i^\top$ to a matrix $\R$ with $\v_1 \v_1^\top$ already in its support increases rank $\Gamma(\R)$ by only one. 
Thus, $\R = \W^{(1)} - \M$ being $K$-sparse with support set $\cI$, $|\cI| = K-1$, means $\R$ is rank-$K$.

The complexity of the Full OMP algorithm is dominated by the computation of the full eigenvector set $\{\v_i\}$ of $\W^{(1)}$ to construct $\V$, which in the worst case is $\cO(N^3)$. 

%Note that in the general case when $\W^{(1)} - \M$ has different eigenvectors corresponding to non-zero eigenvalues than $\cV$, Full OMP algorithm is not globally optimal.Note that the eigenvectors of solution $\Z^* = \W^{(1)} - \R^*$ to \eqref{eq:rankConst} can differ from original $\W^{(1)}$, since by definition eigenvectors are successive orthonormal arguments that minimize the Rayleigh quotient $\min_{\v} \frac{\v^\top \Z^* \v}{\v^\top \v}$ for symmetric and real $\Z^*$.  
%Thus, the eigen-structure of the adjacency matrices slowly evolves over time in the general case.

\vspace{0.05in}
\subsubsection{Fast OMP Algorithm}
\label{subsubsec:fast}

To reduce the computation complexity of Full OMP, we develop a faster variant that does not require full eigen-decomposition of $\W^{(1)}$. 
Instead of dictionary $\cD$ \eqref{eq:setR} containing $3N-3$ rank-1 matrices, we consider a smaller dictionary $\cD'$ that contains $3B-3$ matrices:
\begin{align}
\cD' \triangleq \left\{ ~ \{\v_i \v_i^\top\}_{i=2}^B, ~\{\g_i \g_i^\top\}_{i=2}^B, ~\{\h_i \h_i^\top \}_{i=2}^B ~ \right\} .
\label{eq:setR'}
\end{align}
In words, $\cD'$ is a sub-dictionary of rank-1 matrices that are composed of $B$ eigenvectors of $\W^{(1)}$ corresponding to the $B$ largest eigenvalues.
These $B$ eigenvectors can be computed in linear time, assuming $B \ll N$, via fast numerical linear algebra algorithms, such as \textit{Locally Optimal Block Preconditioned Conjugate Gradient} (LOBPCG) \cite{knyazev01}.

Thus, a simple fast variant is to first compute $B$ eigenvectors $\v_i$ of $\W^{(1)}$ corresponding to the $B$ largest eigenvalues, construct rank-1 matrix dictionary $\cD'$, and follow the flow diagram in Fig.\;\ref{fig:greedy_algo_flowchart}, using $\cD'$ instead of $\cD$.
Assuming $\W^{(1)} - \M$ is sparse---$\W^{(1)}$ specifies a sparse graph while $\M$ is a candidate for sparse $\W^{(2)}$---each inner-product computation has complexity $\cO(N)$.
To ensure $\W^{(1)} - \M$ is sparse, after computing gradient descent $\W^{(2)} - \epsilon \nabla_{\W^{(2)}}$ in Section\;\ref{subsubsection:gradientDescent}, we perform \textit{hard-thresholding} to sparsify $\W^{(2)} - \epsilon \nabla_{\W^{(2)}}$ as input to the proximal mapping. 
Thus, the complexity of the fast greedy algorithm is $\cO(N)$.

\vspace{0.05in}
\subsubsection{Comparison of Full and Fast OMP}
To compare Full OMP and Fast OMP, we generated $10$ synthetic graphs of sizes 
$ N \in \{50, 80, 100, 120, 150, 180, 200, 250, 300, 400\} $.
Each initial adjacency matrix $\W^{(1)}$ was sampled from a Bernoulli model with probability $p=0.1$, then symmetrized with zero diagonal terms. 
We constructed a ground-truth matrix $\W^{(2)}$ by adding a rank $\leq 3$ update to $\W^{(1)}$, ensuring $\W^{(2)} - \W^{(1)}$ is exactly low-rank. 
To simulate a noisy version of the ground truth $\W^{(2)}$, we further added small Gaussian perturbations to obtain $\M$, which serves as a candidate for $\W^{(2)}$. 
We then applied hard-thresholding to $\M$ using $\epsilon = 10^{-3}$ to sparsify it.

\begin{table}[ht!]
\centering
\caption{Runtime (seconds) and Frobenius error \(\|\Z - \W^{(2)}\|_F\)
         for Full vs. Fast Greedy across ten graph sizes.}
\label{tab:fg-fag-time-error}
\begin{tabular}{c|cc|cc}
\hline
\multirow{2}{*}{Size \(N\)} & \multicolumn{2}{c|}{Full} & \multicolumn{2}{c}{Fast} \\
& Time(s) & Error & Time(s) & Error \\
\hline
100& 0.0555& 43.3007& 0.0112& 45.7107\\
250& 0.1762& 85.2888& 0.0223& 88.4196\\
500& 0.6242& 125.7618& 0.0393& 128.4137\\
750& 2.1958& 156.2433& 0.0508& 158.6105\\
1000& 1.9460& 173.9408& 0.0609& 176.0919\\
2500& 20.0321& 282.3022& 0.2410& 283.9948\\
5000& 83.6362& 401.1751& 0.7677& 403.0081\\
\hline
\end{tabular}
\end{table}

% \begin{figure}[ht!]
% \centering
% \includegraphics[width=0.48\textwidth]{Images/fg_fag_comparison.png}
% \caption{Comparison of runtime and estimation error for 
%          Full Greedy vs. Fast Greedy.}
% \label{fig:fg-fag-plot}
% \end{figure}

We applied full eigen-decomposition\footnote{numpy.linalg.eig()} for the full greedy algorithm and LOBPCG\footnote{scipy.sparse.linalg.lobpcg()} for the fast greedy algorithm to compute $B$ extreme eigenvectors where $B = 10$. 
All experiments were performed in Python 3.10.15 (conda-forge) on macOS 26.3.1 with an Apple M5 CPU, using NumPy 1.26.4 and SciPy 1.15.3.
We recorded the runtime (in seconds) and the Frobenius-norm error to measure how accurately each algorithm estimated the ground-truth matrix. 

Table\;\ref{tab:fg-fag-time-error} presents these results. %and Figure\;\ref{fig:fg-fag-plot} visualizes the same data.  
As shown, Full OMP achieves slightly lower reconstruction errors but requires significantly more time when $N$  becomes larger. 
Fast OMP, in contrast, completes each run in under a second for all sizes, offering dramatic speed-ups at the cost of a modest increase in error.

%\vspace{0.05in}
\subsection{Justification of Assumption}
\label{subsec:justify}

We justify the assumption that update matrix $\P = \W^{(2)} - \W^{(1)}$ is a sparse linear combination of rank-1 matrices in dictionary $\cD$ that are outer-products of $\W^{(1)}$'s eigenvectors $\v_i$. 
Specifically, outer-product $\v_1 \v_1^\top$ of eigenvector $\v_1$ corresponding to the largest eigenvalue $\lambda_1$ provides the best rank-1 approximation of $\P$.
$\v_1$ of a non-negative irreducible matrix $\W^{(1)}$ is called the \textit{Perron vector}. 
By the Perron-Frobenius Theorem \cite{johnson12}, it is a strictly positive vector, whose entry magnitudes $|v_{1,k}|$ denote the relative importance of node $k$.
The importance measure here is \textit{eigenvector centrality} \cite{bonacich87}, meaning that node $k$ is important iff its neighbors $l$'s are also important; \ie, node $k$'s importance value (centrality) $s_{k}$ is a weighted linear combination of its neighbors' centralities $s_{l}$'s:
\begin{align}
s_{k} &= \frac{1}{\lambda} \sum_{(k,l) \in \cE} w_{k,l} s_{l}, 
~~~~~
\lambda \s = \W^{(1)} \s .
\end{align}
This implies $\s$ is an eigenvector of $\W^{(1)}$ with eigenvalue $\lambda$. 
Among $\W^{(1)}$'s $N$ eigenvectors, $\v_1$ is the only eigenvector with strictly positive entries.
Hence, $\v_1$ is interpreted as an importance vector, and the outer product $\v_1 \v_1^\top$ emphasizes interaction between important nodes.

Next, we provide a more mathematically rigorous analysis.
%Denote by $\P = \tilde{\W} - \W$ the update matrix.
Given that the eigenvectors $\v_i$'s of $\W^{(1)}$ are orthonormal and span the signal space, we can write update matrix $\P$ as
\begin{align}
\P &= \I \, \P \, \I = \sum_i \v_i \v_i^\top \P \sum_j \v_j \v_j^\top
\nonumber \\
&= \sum_i \sum_j \v_i \underbrace{(\v_i^\top \P \v_j)}_{q_{i,j}} \v_j^\top 
\end{align}
where $\I = \sum_i \v_i \v_i^\top$.
Suppose now that $\P = \alpha \u \u^\top$ is \textit{rank-1}, where vector $\u$ is unit-norm.
Then,
\begin{align}
\v_i^\top \P \v_j = \alpha (\v_i^\top \u) (\u^\top \v_j) .
\end{align}
Thus, the coefficients $q_{i,j}$'s of terms in the double sum are significant only if the projections of update vector $\u$ onto $\v_j$ and $\v_i$ are \textit{both} large. 
Suppose further that $\u$ is a random unit-norm vector. 
By the \textit{Concentration of Measure Phenomenon} (CoMP) \cite{Ledoux2001}, the projection of $\u$ onto $\v_i$ is upper-bounded:
\begin{align}
\text{Pr}(|\v_i^\top \u| > \epsilon) \leq C e^{-c \epsilon^2 N} 
\end{align}
where $C, c$ are constants.
Thus, $\u$ and $\v_i$ are already near-orthogonal as dimension $N$ increases, and the probability of projections of $\u$ onto $\v_i$ and $\v_j$ being both large is much smaller than just projection on $\v_i$.
Thus, the first-order approximation of $\P$ is to consider just the single sum:
\begin{align}
\P \approx \sum_i q_{i,i} \, \v_i \v_i^\top .
\end{align}

Suppose further that update vector $\u$ is only positive in a local graph neighborhood, while the rest of $\u$'s entries are zeros. 
This is the case for a $K$nn spatial graph, where each node is connected only to its $K$-nearest spatial neighbors (who are themselves similar) at all time, and a node $i$ changes its behaviour and becomes more/less similar to its $K$ connected neighbors at $t=2$. 
Then, $\u^\top \v_1$ is large, since $\v_1$ is strictly positive (by Perron-Frobenius Theorem) and smooth.
That means coefficients $\v_1^\top \P \v_i$'s and $\v_i^\top \P \v_1$'s are the next largest.
Thus, the second-order approximation of $\P$ is
\begin{align}
\P &\approx \sum_{i=1}^N q_{i,i} \v_i \v_i^\top + \sum_{i \neq 1} q_{1,i} \v_1 \v_i^\top + q_{i,1} \v_i \v_1^\top
\nonumber \\
&\stackrel{(a)}{=} \sum_{i=1}^N q_{i,i} \v_i \v_i^\top + \sum_{i \neq 1} q_{1,i} \left( \v_1 \v_i^\top + \v_i \v_1^\top \right)
\end{align}
where $(a)$ follows since $q_{1,i} = q_{i,1}$ for symmetric $\P$.  
$\P$ being localized also means that the $0$-order approximation of $\P$ is just the first component, \ie, $q_{1,1} \v_1 \v_1^\top$.

Since $\v_1 \v_i^\top + \v_i \v_1^\top$ is a symmetric rank-2 matrix, we can write it instead using $\g_i$ and $\h_i$ in \eqref{eq:setR2}:
\begin{align}
\g_i \g_i^\top - \h_i \h_i^\top = \v_1 \v_i^\top + \v_i \v_1^\top .
\end{align}
We can conclude that $\P$ can be approximated as a sparse linear combination of $\v_i \v_i^\top$'s, $\g_i \g_i^\top$'s and $\h_i \h_i^\top$'s, which is $\cD$.

%In practice, $\{b_i\}$ are first computed by singular value decomposition (SVD) on $\W^{(1)}$, then $\{a_i\}$ are computed by projecting $\M$ onto rank-1 matrices $\u_i \v_i^\top$ that are outer product of $\W^{(1)}$'s computed singular vectors. The solution to \eqref{eq:proxNuclear} is $\Z^* = \U \S \V^\top$, where $\S$ is diagonal with $S_{i,i} = s_i^*$. 

\subsection{Projection Operator}

To address the third convex but non-smooth term in \eqref{eq:uOpt}, one can easily show that the proximal mapping for the indicator function $\text{I}_W(\W^{(2)})$ is a convex-set projection operator \cite{parikh13}:

\vspace{-0.05in}
\begin{small}
\begin{align}
\text{prox}_{I_W}(\M) &= \text{proj}_{\cC}(\M)  
\\
(\text{proj}_{\cC}(\M))_{i,j} &= 
\left\{ \begin{array}{ll}
0 & \mbox{if}~ i=j 
\label{eq:csProj} \\
0 & \mbox{if}~ i \neq j, M_{i,j} < 0 \\
\frac{M_{j,i} + M_{i,j}}{2} & \mbox{if}~ i \neq j, M_{j,i}, M_{i,j} \geq 0
\end{array}
\right. .
\end{align}
\end{small}

\vspace{-0.1in}
\subsection{Algorithm Summary}

We summarize our PGD algorithm as follows: 
\begin{align}
\W^{(2)} \leftarrow \text{proj}_{\cC} (\text{prox}_{g,\W^{(1)}}(\W^{(2)} - \epsilon \nabla_{\W})) .
\label{eq:PGD}
\end{align}
In words, we first decrease the first term in \eqref{eq:uOpt} by updating $\W^{(2)}$ in the negative gradient direction $-\nabla_{\W}$, where $\epsilon > 0$ is a step size. 
We then perform proximal mapping \eqref{eq:proxRank}.
Finally, we perform convex-set projection \eqref{eq:csProj}.
\eqref{eq:PGD} is executed iteratively until solution $\W^{(2)}$ converges.
Signal $\x_2$ and adjacency matrix $\W^{(2)}$ are alternately optimized until the pair converges.
See Appendix\;\ref{append:convergence} for a convergence proof given a carefully chosen $\epsilon$.

% \vspace{0.05in}
% \begin{enumerate}
% \item Perform gradient descent $\M^1 \leftarrow \W^{(2)} - \epsilon \nabla_{\W}$. 
% \item Perform proximal mapping $\M^2 \leftarrow \text{prox}_{g,\W^{(1)}}(\M^1)$.
% \item Perform projection $\W^{(t)} \leftarrow \text{proj}_{\cC}(\M^2)$.
% \end{enumerate}
% \vspace{0.05in}

To interpolate $\y_t \in \mathbb{R}^{M_t}$ to $\x_t \in \mathbb{R}^{N}$ for each $t \in \{2, \ldots, T\}$, we compute each $\x_t$ iteratively, \ie, compute first $\x_2$ and $\W^{(2)}$ using previous $\W^{(1)}$, then $\x_3$ and $\W^{(3)}$ using previous $\W^{(2)}$, etc.

\section{Algorithm Unrolling}
\label{sec:unrolling}

%\begin{figure*}[t]
\begin{figure}[t]
    \centering
    \includegraphics[width=0.49\textwidth]{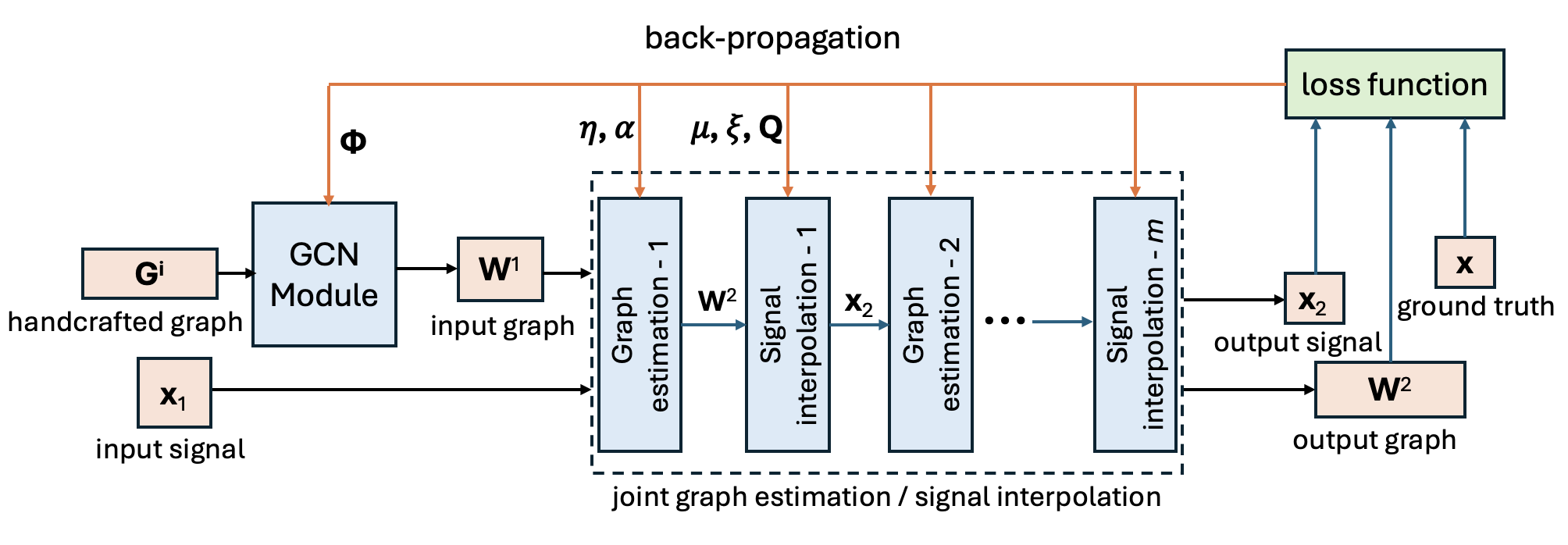}
    \vspace{-0.25in}
    \caption{Neural Network Architecture with GCN and Unrolled Layer.}
    \label{fig:neural_network_architecture}
\end{figure}    
%\end{figure*}

%\vspace{-0.05in}
\subsection{Unrolling of Optimization Algorithm}
\label{subsec:unrolling}

Given the alternating algorithm to compute $\x_2$ and $\W^{(2)}$ developed in the previous section, we unroll its iterations into layers to construct a lightweight feed-forward network amenable to end-to-end data-driven parameter tuning.
The tunable parameters in our algorithm are: 
\begin{enumerate}
\item prior weights $\mu$, $\xi$ and $\eta$ in objective \eqref{eq:jointForm}; 
\item two CG parameters $\kappa$ and $\beta$ corresponding to conjugate gradient descent step size and momentum used to solve linear system \eqref{eq:linSys1} for $\x_2^*$ given $\W^{(2)}$; 
\item the step size $\epsilon$ for gradient descent in the PGD algorithm \eqref{eq:PGD} for $\W^{(2)}$; and
\item the diagonal entries $Q_{i,i}$ in matrix $\Q$ of the fidelity term. 
\end{enumerate}
Note that these parameters are tuned per unrolled layer. 

Based on the alternating optimization algorithm described in Section~\ref{Alternating Optimization Algorithm}, our unrolled module consists of two types of layers. 
The first type optimizes the adjacency matrix $\W^{(2)}$ assuming the signal $\x_2$ is fixed, while the second type optimizes $\x_2$ assuming $\W^{(2)}$ is fixed.
See Fig.\;\ref{fig:neural_network_architecture} for an illustration.

\subsection{Learning an Initial $\W^{(1)}$ via GCN}
\label{subsec:initialGraph}

To construct an initial adjacency matrix $\W^{(1)}$, we first use handcrafted features $\f_i$ per node (\eg, location and land value) to build an initial graph $\cG^i$ based on feature distance, \ie, $w_{i,j} = \exp (- \|\f_i - \f_j\|^2_2)$.
We then feed $\cG^i$ and the node feature matrix $\X$ (which includes $\f_i$s) into a shallow \textit{Graph Convolutional Network} (GCN) \cite{kipf17,Wu21}. 

Figure\;\ref{fig:shallow_gcn-plot} shows the architecture of the GCN. 
A GCNConv layer aggregates feature information from its neighbors. After feature aggregation, a non-linear transformation (ReLU) is applied to the resulting outputs. By stacking multiple layers, the final representation of each node receives relevant features from a larger neighborhood, and the output graph specified by $\W^{(1)}$ better captures the data's structure, providing more meaningful features.
Subsequently, we use $\W^{(1)}$ as the initial graph for our unrolled alternating algorithm. 
Note that we use a very shallow GCN, so the number of parameters is small, which makes it suitable for data-starved scenarios.

\begin{figure}[ht!]
\centering
\includegraphics[width=0.4\textwidth]{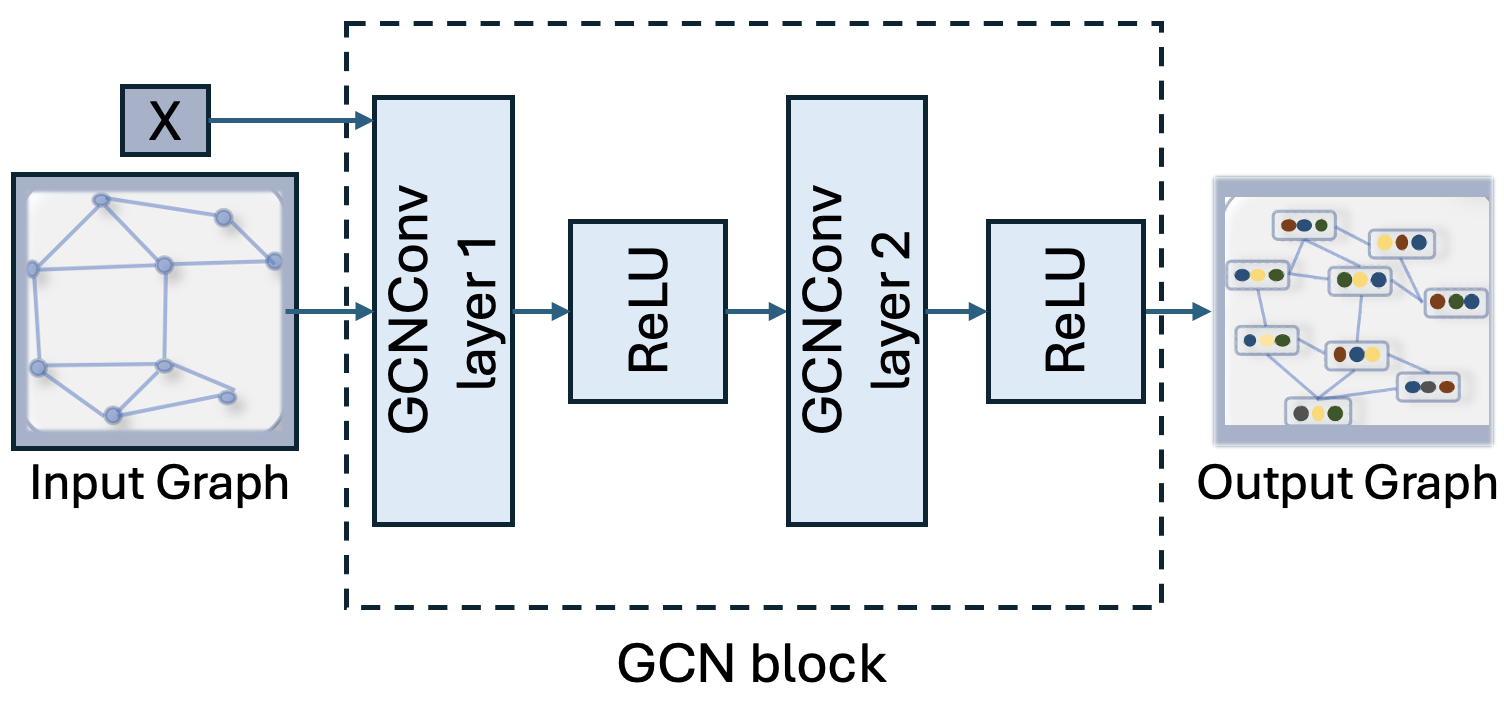}
\caption{Shallow GCN Architecture}
\label{fig:shallow_gcn-plot}
\end{figure}

\section{Experiments}
\label{sec:results}

\vspace{-0.1in}
\subsection{Datasets}
\subsubsection{Quarterly Farmland Sales Dataset}

The first dataset, obtained from the USDA’s National Agricultural Statistics Service (NASS)\footnote{\url{https://www.nass.usda.gov/}}, contains farmland sales data from various counties in Iowa, US, spanning the fiscal quarters from 2019 to 2023. 
It includes attributes such as gross price, gross acres, tillable acres, sale date, and soil rating. 
Each data point represents a farmland sale in dollars per acre. 
Iowa was chosen due to its extensive availability of county-level data in every quarter. 
%allowing for validation of the low-rank update assumption across consecutive quarters. 
Given that soil rating significantly affects land value, an adjusted price metric is employed, derived through linear regression between the actual price and the soil rating. 

% To illustrate the linear relationship more clearly, Figure \ref{fig:sales_data_regression} displays a scatter plot of farmland average price per acre versus soil rating, alongside the best-fit line obtained from the linear regression used to create the adjusted price metric.

\subsubsection{Daily Temperature Dataset} 

The second dataset, obtained from the Meteostat Python library\footnote{\url{https://dev.meteostat.net/}}, contains daily maximum temperature observations from $107$ weather stations in Ontario, Canada. The dataset includes attributes such as station ID, station name, geographic coordinates (latitude and longitude), and daily maximum temperature readings for each station. 
%To capture temperature trends for this experiment, 
The dataset spans from June 1, 2020 to July 31, 2020. 
Each data point represents the daily maximum temperature of a station, sourced from a mixture of public and governmental weather data interfaces aggregated by Meteostat.

\subsubsection{Traffic Prediction Dataset} 

The third dataset, obtained from the PyTorch Geometric Temporal Datasets\footnote{\href{https://pytorch-geometric-temporal.readthedocs.io/en/latest/modules/dataset.html\#torch_geometric_temporal.dataset.metr_la.METRLADatasetLoader}{METRLADatasetLoader}}, contains traffic readings from 207 detectors on highways in Los Angeles County. 
The dataset contains observations every $15$ minutes between March 2012 to June 2012.
Each data point captures the speed of traffic during the day at each sensor.

\vspace{-0.1in}
\subsection{Experimental Setup}

We compare our proposed joint graph learning / signal interpolation framework with three representative baselines.
The first two are exclusively time‑varying graph learning methods; after estimating the adjacency at all time instants, we recover the missing signal values by solving the linear system in \eqref{eq:linSys}. 
As done in our conference version \cite{bagheri24}, 
\texttt{Tikhonov} and \texttt{sparse} are chosen as graph learning schemes, while \texttt{SVD} and \texttt{OTS} are omitted due to their slow execution speed.

The third baseline \texttt{NestedDAU} is a deep learning-based signal interpolation method that operates on a fixed graph topology, which remains unchanged over time  \cite{nakahama22}. 
We constructed a feature-distance-based adjacency matrix using the nodes available at the initial time step and applied \texttt{NestedDAU} for interpolation.

We summarize the methods as follows:
\begin{enumerate}
\item \texttt{Tikhonov} \cite{kalofolias17} constrains changes in a graph across time by penalizing the Frobenius norm of the difference between consecutive adjacency matrices.
%, ensuring that variations in the graph’s structure remain minimal. 
%Such a smoothness prior is particularly effective for applications where the graph topology changes gradually, as is often the case in networks such as biological interaction graphs, where connections tend to shift incrementally rather than abruptly. By prioritizing stability, \texttt{Tikhonov} provides a robust framework for modeling time-series data on graphs with consistent structural patterns, though it may struggle to capture rapid or significant topological shifts.

\item Assuming  that most connections remain static, \texttt{sparse} \cite{yamada19} 
%This method extends the \texttt{Tikhonov} approach by incorporating a sparsity constraint to better model graphs where only a small number of edges change between time steps. Rather than assuming uniform smoothness across all edges, 
applies a sparsity prior to the difference between consecutive adjacency matrices. 
%which promotes parsimonious updates and enhances computational efficiency. The resulting variant, known as \texttt{TVGL-sparse}, is particularly well-suited for scenarios such as sensor networks, where connectivity changes are driven by occasional events.

\item \texttt{NestedDAU} \cite{nakahama22} adopts the Plug-and-Play ADMM (PnP-ADMM) framework by unrolling it into a deep neural network tailored for GSP. 
Each unrolled layer splits the computation into two stages: an inverse step, which reconstructs partially observed signal nodes using available data, and a denoising step, which imposes a graph-structured prior to regularize the solution. 
%The denoising component is itself an unrolled ADMM procedure designed specifically for graph signals, resulting in a nested architecture that captures intricate signal dependencies while maintaining computational efficiency. 
%Through end-to-end training of layer-wise parameters, \texttt{NestedDAU} strikes a balance between fidelity to the observed data and regularization, making it particularly effective in settings with limited or noisy training data, such as dynamic graph-based sensor systems or incomplete network datasets. However, in contrast to the time-varying graph model introduced in this work, 
\texttt{NestedDAU} relies on a single, fixed adjacency matrix defined at the initial time step.
\end{enumerate}

For the farmland sales dataset (2021-Q2 to 2023-Q2), data from 22 Iowa counties are selected across eight consecutive quarters. The signals are min-max normalized, and $11$ out of $22$ observations per quarter are randomly omitted (50\% missing rate). Gaussian noise with mean $0$ and standard deviation $0.9$ is then added, with the goal of recovering the complete signals for all these quarters.

For the daily temperature dataset (June 2–21, 2020), daily maximum temperature readings from $107$ weather stations in Ontario are used over $20$ consecutive days. Each day’s measurements are normalized, $40$ out of $107$ observations are randomly removed ($\sim35$\% missing rate), and the same Gaussian noise is introduced. The task is to reconstruct the full daily temperature signals.

Finally, the 5-min traffic dataset originally contains measurements from March 2012 to June 2012 for $207$ sensors; Monday, May 28, 2012 was selected as a representative day with the readings between 5am and 9pm used here. Each observation’s signals are normalized, $72$ out of $207$ observations are omitted ($\sim35$\% missing rate), and Gaussian noise is added. The aim is to estimate the complete speed output signals over this 16-hour window.

For each of the datasets, the conference version \cite{bagheri24} uses parameters \((\mu = 0.1, \alpha = 0.1, \eta = 0.1, \gamma = 0.001)\). Both \texttt{Tikhonov} and \texttt{TVGL-sparse}, use \((\alpha = 0.1, \gamma = 0.001)\). For \texttt{NestDAU}, the number of unrolled layers $P$ is set to 5, and the number of channels is $1$. 
In our proposal (journal), initial parameters are set to ($\mu = 0.1$, $\xi = 10^{-5}$, and $\eta = 0.25$) in objective \eqref{eq:jointForm}. The step size $\epsilon$ for gradient descent in the PGD algorithm \eqref{eq:PGD} is 0.1,  and the diagonal entries $Q_{i,i}$ in matrix $\Q$ are initialized based on the number of observations. 
The same initial settings are applied to all three datasets for both our proposal and \texttt{NestDAU}, since they both learn parameters via algorithm unrolling for each dataset. 
%thereby further improving interpolation performance.

% The initial learning rate is set to $0.01$, and the training proceeds for $400$ epochs.

\vspace{-0.1in}
\subsection{Experimental Results}
\label{sec:results_timevarying2}

Fig.\, \ref{fig:rmse_comparison_landsales} shows the performance of five competing methods on the farmland sales dataset across eight quarters. Our method achieves the lowest error metrics in nearly every quarter, demonstrating its consistent improvements over the others.

\begin{figure}[ht!]
\vspace{-0.1in}
\centering
\includegraphics[width=0.5\textwidth]{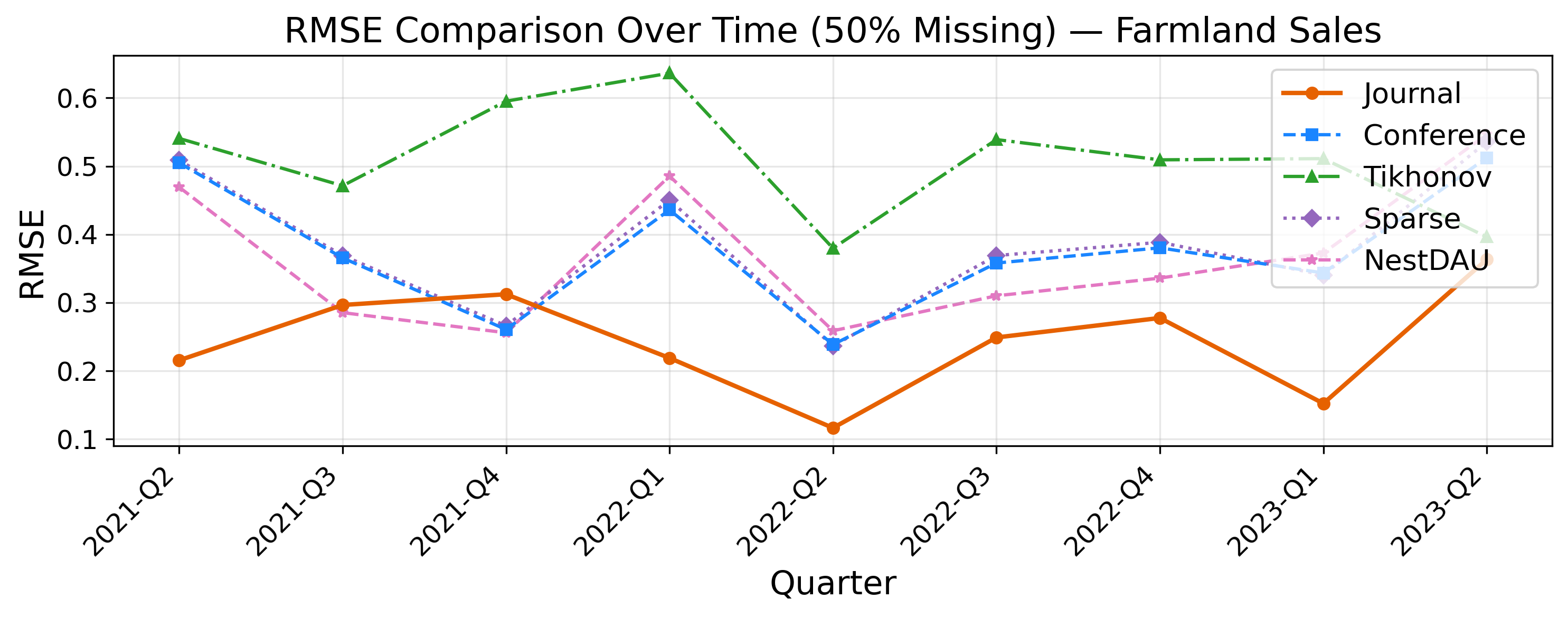}
\vspace{-0.2in}
\caption{RMSE Comparison of Five Interpolation Methods by Quarter (50\% Missing Rate) -  Farmland Sales Dataset}
\label{fig:rmse_comparison_landsales}
\end{figure}

Fig.\,\ref{fig:rmse_comparison_weather} illustrates the performance of five competing methods on the daily temperature dataset from 2020-06-02 to 2020-06-21. 
The proposed method generally achieves the lowest errors overall and on most days. 
\texttt{Tikhonov} tracks it closely from 2020-06-07 until the end after starting with much higher error.

\begin{figure}[ht!]
\vspace{-0.1in}
\centering
\includegraphics[width=0.5\textwidth]{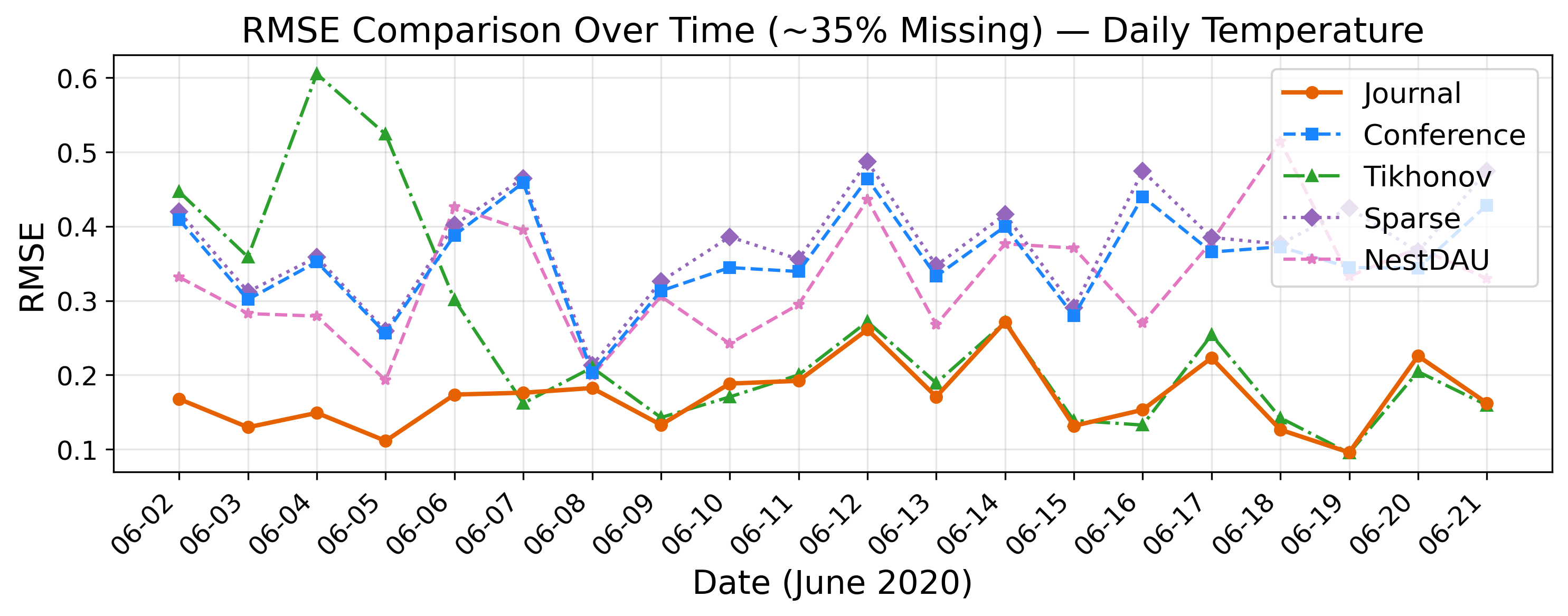}
\vspace{-0.2in}
\caption{RMSE Comparison of Five Interpolation Methods by Quarter ($\sim35$\% Missing Rate) -Temperature Data}
\label{fig:rmse_comparison_weather}
\end{figure}

Fig.\, \ref{fig:rmse_comparison_traffic} illustrates the performance of the five competing methods on the 5-min traffic dataset. The proposed method generally achieves the lowest error metrics across most time intervals, while the conference version \cite{bagheri24} and \texttt{Tikhonov} produce similar per-step RMSE.

Collectively, the results across the three datasets demonstrate that the journal method consistently delivers superior performance compared to competing schemes. 
\texttt{Tikhonov} performs well for the weather and traffic datasets after higher error during the initial time instants. 
This coincides with simple rank-1 updates by the journal method, indicating that minimal graph changes are needed in later time instants. 
%In this scenario, the journal method and \texttt{Tikhonov} are effectively solving GLR interpolation. 
%When the journal method was compared by dropping gi and hi terms, the method proposed performed better for the traffic dataset while weather dataset required twice as many atoms to achieve equivalent performance.

\begin{figure}[ht!]
\vspace{-0.1in}
\centering
\includegraphics[width=0.5\textwidth]{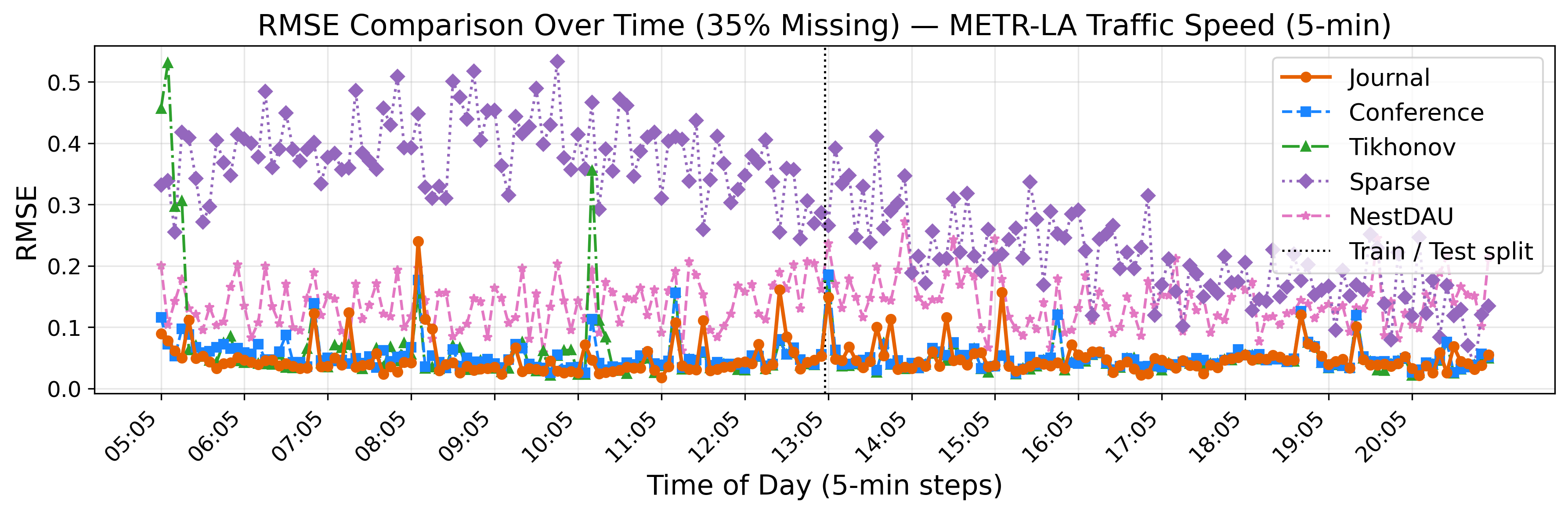}
\vspace{-0.2in}
\caption{RMSE Comparison of Five Interpolation Methods by 5-min Intervals (35\% Missing Rate) - Traffic Data}
\label{fig:rmse_comparison_traffic}
\end{figure}

Finally, Table \ref{tab:rmse_results_all} summarizes the average RMSE interpolation error for all methods across all datasets. 
The averages were computed over 10 quarters for the farmland sales dataset, 20 days for the temperature dataset, and 16 hours for the traffic dataset. 
Overall, the journal method exhibits the lowest average error, highlighting its effectiveness in adapting to diverse data characteristics and temporal scales while maintaining high interpolation accuracy.

For the traffic dataset, we observe that compared to the case when dictionary $\cD$ is limited to only outer-products $\v_i \v_i^\top$'s, full $\cD$ containing $\g_i \g_i^\top$'s and $\h_i \h_i^\top$'s also improves RMSE by $2.2\%$. 
This demonstrates the usefulness of the full dictionary when the time-varying graph is sufficiently large and requires more complex combinations of rank-1 matrices over time beyond simple $\v_i \v_i^\top$'s, as discussed in Section\;\ref{subsec:justify}.

\begin{table}[h!]
\scriptsize
\centering
\caption{Comparison of Average RMSE Errors for Interpolation Methods across Datasets (Lowest Errors in Bold).}
\begin{tabular}{lccccc}
\toprule
Dataset   & Conference & Journal & Tikhonov & Sparse & NestDAU \\
\midrule
Land Sales  & 0.3777& \textbf{0.2446}& 0.5088& 0.3849& 0.3689\\
Weather   & 0.3569& \textbf{0.1713}& 0.2490& 0.3770& 0.3299\\
Traffic& 0.0515& \textbf{0.0489}& 0.0575& 0.2948& 0.1405\\
\bottomrule
\end{tabular}
\label{tab:rmse_results_all}
\end{table}

\vspace{-0.1in}
\subsection{Ablation Study}

To evaluate the impact of key hyperparameters, we conducted ablation studies on the farmland sales dataset, reporting the average RMSE over eight quarters. Unless otherwise specified, the default configuration uses a learning rate of $0.01$, $400$ training epochs, $4$ unrolled layers in the interpolation network, and $8$ hidden layers in the GCN module.

We first evaluated the benefit of updating the graph at each time step (dynamic) versus using a single static graph learned from 2021-Q2 only. The static graph was trained on all available observations from quarter 1, and then used unchanged for interpolation across all eight quarters. In contrast, the dynamic approach updates the adjacency matrix at each quarter via our unrolled algorithm. The static-graph setting yielded an average RMSE of 0.2857, whereas the dynamic-graph approach reduced the RMSE to 0.2446, confirming the value of modeling time‑varying connectivity.

Next, we fixed all other parameters and varied the learning rate. Table \ref{tab:lr_rmse} reports the resulting RMSE after 400 epochs. A learning rate of 0.01 achieved the lowest error, while both smaller and larger rates led to slower convergence or slight over‑/under‑fitting.

%\vspace{-0.1in}
\begin{table}[ht]
  \centering
  \caption{Mean RMSE for Different Learning Rates}
  \label{tab:lr_rmse}
  \begin{tabular}{lc}
    \toprule
    Learning Rate (lr) & Mean RMSE \\
    \midrule
    0.0001  & 0.2502\\
    0.01    & \textbf{0.2446}\\
    0.1     & 0.2738\\
    1& 0.2938\\

    \bottomrule
  \end{tabular}
\end{table}

We then varied the number of unrolled layers \(l\) while holding other settings constant. Table~\ref{tab:unrolled_layers} shows both the RMSE and the average wall‑clock time per run. The results show \(l=4\) with the lowest RMSE (0.2446). Increasing the number of layers above that increases RMSE.

%\vspace{-0.05in}
\begin{table}[ht]
  \centering
  \caption{Performance for Different Numbers of Unrolled Layers (400 epochs)}
  \label{tab:unrolled_layers}
  \begin{tabular}{lcc}
    \toprule
    Unrolled Layers & Time (s) & RMSE \\
    \midrule
    2   &  2.76& 0.2486\\
    4   &  4.63& \textbf{0.2446}\\
    8&  8.31& 0.2753\\
    16& 15.74& 0.2781\\
    \bottomrule
  \end{tabular}
\end{table}

We examined the depth of the GCN feature extractor by varying the number of hidden layers \(H\). Table~\ref{tab:gcn_hidden_layers} reports the RMSE results using the daily temperature dataset. The test was run on the farmland sales dataset but resulted in equal RMSE due to lack of diversity in the dataset for the GCN.
Shallower networks (\(H=4\)) underfit, and deeper networks (\(H\ge16\)) again show increased error, suggesting overfit. 
\(H=8\) strikes the best balance between model capacity and generalization.

\begin{table}[ht]
  \centering
  \caption{RMSE for Different GCN Hidden Layers}
  \label{tab:gcn_hidden_layers}
  \begin{tabular}{lc}
    \toprule
    GCN Hidden Layers & RMSE \\
    \midrule
    4   & 0.1727\\
    8   & \textbf{0.1713}\\
    16  & 0.1717\\
    32  & 0.1720\\
    \bottomrule
  \end{tabular}
\end{table}

Due to the nature of the rank weight loss being discrete since rank is an integer quantity, learning the parameter results in unstable gradients in the unrolled algorithm. 
We tested how variations in $\eta$ affected RMSE to determine an appropriate initial parameter. Additionally, the mean rank per layer was computed to show how many atoms were used in the update. 
The expectation is that as $\eta$ gets smaller, the rank constraint will be weaker and thus a high rank update. This is shown in the tests with a rank-1 update occuring when $\eta = 0.25$.

\vspace{-0.05in}
\begin{table}[ht]
  \centering
  \caption{RMSE for Different $\eta$  (400 epochs)}
  \label{tab:rank_weight}
  \begin{tabular}{lcc}
    \toprule
    $\eta$& RMSE& Mean Rank/Layer\\
    \midrule
    0.01&  0.2600& 9.56\\
    0.05&  0.2606& 5.25\\
    0.1&  0.2697& 6.86\\
    0.25& 0.2446& 1.0\\
    0.5& \textbf{0.2444}& 1.0\\
    1.0& 0.2449& 1.0\\
    \bottomrule
  \end{tabular}
\end{table}

\section{Conclusion}
\label{sec:conclude}
To model slowly time-varying graphs for scenarios where pairwise similarities between nodes change over time, we propose a low-rank model---the difference between two consecutive adjacency matrices $\W^{(2)} - \W^{(1)}$ is low-rank. 
We formulate a joint signal interpolation / graph learning problem, where signal $\x_2$ and adjacency matrix $\W^{(2)}$ at time $t=2$ are alternately optimized, given $\x_1$ and $\W^{(1)}$ at $t=1$.
When $\x_2$ is fixed, we propose a proximal gradient descent (PGD) algorithm to compute $\W^{(2)}$, where the proximal mapping for the rank term $\Gamma(\W^{(2)} - \W^{(1)})$ can be efficiently computed using an orthogonal matching pursuit (OMP) algorithm considering only outer-products of $\W^{(1)}$'s extreme eigenvectors. 
%assuming that the eigenvector set of $\W^{(1)}$ contains the eigenvectors corresponding to non-zero eigenvalues of $\W^{(2)} - \W^{(1)}$. 
We unroll our algorithm into layers to compose a lightweight network for data-driven parameter learning.
Experiments show that our joint optimization achieves better signal interpolation results than existing time-varying graph models.

\vspace{-0.1in}
\appendices

\section{Least-Square Solution}
\label{append:LS_soln}

We prove that the least-square problem \eqref{eq:LS} has linear system \eqref{eq:LS_soln} as solution.
We first rewrite the Frobenius norm in \eqref{eq:LS} in vector form, using definitions in \eqref{eq:LS_defn}:
\begin{align}
\left\|
\V^\top \a - \m
\right\|^2_2 &= (\V^\top \a - \m)^\top (\V^\top \a - \m) 
\nonumber \\
&= \a^\top \V \V^\top \a - 2 \a^\top \V \m + \m^\top \m .
\end{align}
Taking the derivative w.r.t. $\a$ and setting it to $0$:
\begin{align}
2 \V \V^\top \a - 2 \V \m = 0
\end{align}
We rearrange terms to get \eqref{eq:LS_soln}.

\section{Proof of Optimality}
\label{append:optimality}

We prove Theorem\;\ref{thm:optimality} here.
The proof follows similar arguments as the known \textit{mutual incoherence condition} (MIC) from sparse coding \cite{Tropp2004}, where \textit{mutual incoherence} is defined as
\begin{align}
\rho(\cD) = \max_{\r_i, \r_j \in \cD, i \neq j} \frac{ \langle \r_i, \r_j \rangle}{\|\r_i\| \|\r_j\|} .
\end{align}

\begin{proof}

Note first that the inner-products between rank-1 components in $\cD \cup \{ \v_1 \v_1^\top\}$ can be easily upper-bounded:
\begin{align}
\langle \v_i \v_i^\top,  \v_j \v_j^\top \rangle &= \left\{ \begin{array}{ll}
1 & \mbox{if}~ i = j \\
0 & \mbox{o.w.}
\end{array} \right.
\nonumber \\
\langle \g_i \g_i^\top,  \g_j \g_j^\top \rangle = \langle \h_i \h_i^\top,  \h_j \h_j^\top \rangle &= \left\{ \begin{array}{ll}
1 & \mbox{if}~~ i = j \\
\frac{1}{4} & \mbox{o.w.}
\end{array} \right.
\nonumber \\
\langle \v_i \v_i^\top, \g_j \g_j^\top \rangle = \langle \v_i \v_i^\top, \h_j \h_j^\top \rangle &= \left\{ \begin{array}{ll}
\frac{1}{2} & \mbox{if}~i=1 ~\mbox{or}~i=j \\
0 & \mbox{o.w.}
\end{array} \right.
\nonumber \\
\langle \g_i \g_j^\top, \h_j \h_j^\top \rangle &= \left\{ \begin{array}{ll}
0 & \mbox{if}~ i = j \\
\frac{1}{4} & \mbox{o.w.}
\end{array} \right.
\end{align}
Since all $\r_i \in \cD \cup \{ \v_1 \v_1^\top \}$ have unit-norm, \ie, $\|\r_i \| = \sqrt{\langle \r_i, \r_i \rangle} = 1$, we conclude that $\rho(\cD) = \frac{1}{4}$ for $i > 1, i \neq j$.

Denote by $\s_t$ the \textit{residual} after $t$ iterations of Full OMP. 
Denote by $m$ an index in true support index set $\cI$ and $n$ an index in complementary index set $\cI^c$.
OMP computes the correct support $\cI$ if the following inequality holds for every $m \in \cI, n \in \cI^c$ at every iteration $t \in \{1, \ldots, K-1\}$:
\begin{align}
\max_{\r_m \in \left\{ \substack{
\v_m \v_m^\top, \g_m \g_m^\top,
\\
\h_m \h_m^\top
} \right\} } |\langle \s_t, \r_m \rangle| > \max_{\r_n \in \left\{ \substack{
\v_n \v_n^\top, \g_n \g_n^\top,
\\
\h_n \h_n^\top
} \right\} } |\langle \s_t, \r_n \rangle| .
\end{align}
Given $\v_1 \v_1^\top$ is selected into the support initially, residual $\s_t$ is always orthogonal to $\v_1 \v_1^\top$. 
There are at most $K-1$ additional indices to select into the support at iteration $t$.
Denote by $\cI_t$ a subset of indices remaining to be selected into the support at iteration $t$, \ie, $\cI_t \subseteq \cI$.
Thus, $|\langle \s_t, \r_m \rangle|$, $m \in \cI$, is

\vspace{-0.05in}
\begin{small}
\begin{align}
&= \left| \langle \sum_{i \in \cI_t} a_i \v_i \v_i^\top + b_i \g_i \g_i^\top + c_i \h_i \h_i^\top, \r_m \rangle \right| 
\nonumber \\
&=\left| \sum_{i \in \cI_t} a_i \langle \v_i \v_i^\top, \r_m \rangle + b_i \langle \g_i \g_i^\top, \r_m \rangle + c_i \langle \h_i \h_i^\top, \r_m \rangle \right|
\nonumber \\
&\stackrel{(a)}{\geq} |a_m| - \sum_{i \in \cI_t, i \neq m} |a_i| \, |\langle \v_i \v_i^\top, \r_m \rangle | - \sum_{j \in \cI_t} |b_j| \, |\langle \g_j \g_j^\top, \r_m \rangle | 
\nonumber \\
& ~~~ - \sum_{j \in \cI_t} |c_j| \, |\langle \h_j \h_j^\top, \r_m \rangle| 
~\stackrel{(b)}{\geq}~ |a_m| - \frac{|b_m|}{2} - \frac{|c_m|}{2}
\label{eq:lower_bound_correct}
\end{align}
\end{small}\noindent
where in $(a)$ we assume $\r_m =\v_m \v_m^\top$.
In $(b)$, $|\langle \g_m \g_m^\top, \r_m \rangle| = |\langle \h_m \h_m^\top, \r_m \rangle| = \frac{1}{2}$, and $|\langle \v_j \v_j^\top, \r_m \rangle| = |\langle \g_j \g_j^\top, \r_m \rangle| = |\langle \h_j \h_j^\top, \r_m \rangle| = 0$ for $j \neq m$. 
A similar lower bound can be derived for $\r_m = \g_m \g_m^\top$ as
\begin{align}
&\stackrel{(c)}{\geq} |b_m| - \frac{|a_m|}{2} - \frac{1}{4} \sum_{j \in \cI_t, j \neq m} \left( |b_j| + |c_j| \right) .
\end{align}
In $(c)$, $|\langle \v_m \v_m^\top, \g_m \g_m^\top \rangle| = \frac{1}{2}$ and $|\langle \v_j \v_j^\top, \g_m \g_m^\top \rangle| = 0$ for $j \neq m$, $|\langle \g_m \g_m^\top, \h_m \h_m^\top \rangle| = 0$, and $|\langle \g_j \g_j^\top, \g_m \g_m^\top \rangle| = |\langle \g_j \g_j^\top, \h_m \h_m^\top \rangle| = \frac{1}{4}$ for $j \neq m$.
For $\r_m = \h_m \h_m^\top$, we write
\begin{align}
&\geq |c_m| - \frac{|a_m|}{2} - \frac{1}{4} \sum_{j \in \cI_t, j \neq m} \left( |b_j| + |c_j| \right) .
\end{align}

For $|\langle \s_t, \r_n \rangle|$, $n \in \cI^c$, we write
\begin{align}
&= \left| \sum_{i \in \cI_t} a_i \langle \v_i \v_i^\top, \r_n \rangle + b_i \langle \g_i \g_i^\top, \r_n \rangle + c_i \langle \h_i \h_i^\top, \r_n \rangle \right|
\nonumber \\
&\leq \sum_{i \in \cI_t}  |a_i|  |\langle \v_i \v_i^\top, \r_n \rangle| + |b_i| |\langle \g_i \g_i^\top, \r_n \rangle| + |c_i| |\langle \h_i \h_i^\top, \r_n \rangle| .
\nonumber \\
&\stackrel{(d)}{\leq} 0
\label{eq:upper_bound_incorrect}
\end{align}
where in $(d)$ we assume $\r_n = \v_n \v_n^\top$, and $| \langle \v_i \v_i^\top, \r_n \rangle | = |\langle \g_i \g_i^\top, \r_n \rangle| = |\langle \h_i \h_i^\top, \r_n \rangle| = 0$ for $i \neq n$. 
If $\r_n = \g_n \g_n^\top$, then the upper bound is
\begin{align}
\stackrel{(e)}{\leq} \frac{1}{4} \sum_{i \in \cI_t} |b_i| + |c_i| .
\end{align}
In $(e)$, $|\langle \v_i \v_i^\top \g_n \g_n^\top \rangle| = 0$ for $i \neq n$, and $| \langle \g_i \g_i^\top \g_n \g_n \rangle| = | \langle \h_i \h_i^\top \g_n \g_n \rangle| = \frac{1}{4}$.
Same upper bound can be derived if $\r_n = \h_n \h_n^\top$.

We require the inner-product lower bound for the support component $\r_m$ to be always larger than the inner-product upper bound for the non-support component $\r_n$.
We thus conclude one of the following two inequalities must be true:

\vspace{-0.05in}
\begin{small}
\begin{align}
|a_m| &> \frac{|b_m|}{2} + \frac{|c_m|}{2} + \frac{1}{4} \sum_{i \in \cI_t}  \left( |b_i| + |c_i| \right)
\\
\max \left(|b_m|,\, |c_m|\right) &> \frac{|a_m|}{2} - \frac{|b_m| + |c_m|}{4} + \frac{1}{2} \sum_{i \in \cI_t}  \left( |b_i| + |c_i| \right)
\end{align}
\end{small}\noindent
The sums on the right hand side are largest when all terms are counted in full support set $\cI$.
Thus, the condition is most stringent when $\cI_t = \cI$, leading to \eqref{eq:opt_condition1} and \eqref{eq:opt_condition2}.

Thus, Full OMP identifies the correct support $\cI$ while the coefficients are optimized in a least-square sense via \eqref{eq:LS}.
We can therefore conclude that Full OMP computes an optimal solution $\Z^o$ to \eqref{eq:rankConst} for $r=K$.
By Lemma\;\ref{lemma:duality2}, $\Z^o$ is also a solution to \eqref{eq:proxRank} for a weight parameter value $\eta^*$ such that an optimal solution $\Z^*$ to \eqref{eq:proxRank} results in $\Gamma(\Z^* - \W^{(1)}) = K$.
%
%Given solutions $\{\Z^o\}$ to \eqref{eq:rankConst} for increasing $r$, the full greedy algorithm finds the best one $\Z^o$ among $\{\Z^o\}$ that minimizes \eqref{eq:proxRank} for a given $\eta$, which by Corollary\;\ref{corollary:superset} is also the optimal solution to \eqref{eq:proxRank}.  
\end{proof}

\section{Proof of Convergence}
\label{append:convergence}

Replacing the constraints with the indication function \eqref{eq:ind} in the objective, the optimization \eqref{eq:jointForm} for $\W^{(2)}$ and $\x_2$ can be rewritten in an unconstrained form:

\vspace{-0.1in}
\begin{small}
\begin{align}
\min_{\W^{(2)}, \x_2} & (\y_2 - \H^{(2)} \x_2)^\top \Q (\y_2 - \H^{(2)} \x_2) + \xi \, \|\x_2 - \x_1 \|^2_2 
\nonumber \\
& + \mu \, \text{Tr}((\text{diag}(\W^{(2)}\1 - \W^{(2)}) \x_2 \x_2^\top) 
\nonumber \\
& + \eta \, \Gamma (\W^{(2)} - \W^{(2)}) + \I_W(\W^{(2)}) .
\label{eq:append_obj}
\end{align}
\end{small}\noindent
During our optimization, when $\W^{(2)}$ is fixed, solving for $\x_2^*$ (the \textit{optimal} solution for fixed $\W^{(2)}$) via linear system \eqref{eq:linSys1} means that the sum of the first, second, and third terms is monotonically non-increasing w.r.t. to previous solution $\x_2$.

Conversely, when $\x_2$ is fixed, the PGD algorithm \eqref{eq:PGD} computes a new $\W^{(2)*}$, ensuring that the sum of the third, fourth, and fifth terms in \eqref{eq:append_obj} is monotonically non-increasing with respect to the previous solution \(\W^{(2)}\). 
Specifically, according to \cite{parikh13}, if the gradient $\nabla_{\W}$ is \textit{Lipschitz continuous} with constant $L$, then PGD \eqref{eq:PGD} converges at a rate of $\cO(1/k)$ when a step size $\epsilon \in (0, 1/L]$ is employed. 
%Further, \cite{combettes11} notes that convergence is maintained for step sizes smaller than \(2/L\), although for step sizes larger than \(1/L\), the method ceases to be a \textit{majorization-minimization} approach. 
$L$ can be upper-bounded by the largest eigenvalue of the Hessian matrix.

Since all terms in \eqref{eq:append_obj} are lower-bounded by zero, the objective function is also lower-bounded. 
Thus, the alternating algorithm does not oscillate and converges to a local minimum in a finite number of steps.

% On the other hand, when $\x_2$ is fixed, our proposed PGD algorithm \eqref{eq:PGD} to compute a new $\W^{(2)*}$ means that the sum of the third, fourth and fifth terms is monotonically non-increasing w.r.t. previous solution $\W^{(2)}$. 
% Since all the terms in \eqref{eq:append_obj} is lower-bounded by zero---and thus the objective in \eqref{eq:append_obj}---this implies our alternating algorithm does not oscillates and converges to a local minimum in finite number of steps.

\bibliographystyle{IEEEtran}
\bibliography{ref2}

\end{document}